\documentclass[12pt]{article}
\usepackage{graphicx}
\usepackage{geometry}
\usepackage{arydshln}
\usepackage{enumitem}
\usepackage{float}
\usepackage{hyperref}
\usepackage{xcolor}
\usepackage{authblk}
\usepackage{multirow}
\usepackage{caption}
\usepackage{subfig}
\usepackage{indentfirst}
\usepackage{lscape}

\usepackage{amsmath, amsthm, amsfonts}
\newtheorem{thm}{Theorem}[section]
\newtheorem{cor}[thm]{Corollary}

\theoremstyle{definition}
\newtheorem{defn}[thm]{Definition}
\theoremstyle{remark}

\theoremstyle{assumption}

\newcommand{\abs}[1]{\left\vert#1\right\vert}

\renewcommand{\footnote}{\endnote}  
\title{\textbf{\textcolor{blue}{Black-Litterman Asset Allocation under Hidden Truncation Distribution}}}
\author{\normalsize Jungjun Park\\Department of Economics \\St. Lawrence University, USA\\jpark@stlawu.edu
	\and \normalsize Andrew L. Nguyen\\Department of Mathematics \\California State University Fullerton, USA\\anguyen@fullerton.edu}
\date{\today }
\date{}

\begin{document}
\maketitle

\pdfoutput=1 
\noindent 
In this paper, we study the Black-Litterman (BL) asset allocation model \textcolor{blue}{(Black $\&$ Litterman, 1990)} under the hidden truncation skew-normal distribution \textcolor{blue}{(Arnold $\&$ Beaver, 2000)}. In particular, when returns are assumed to follow this skew normal distribution, we show that the posterior returns, after incorporating views, are also skew normal. By using Simaan's three moments risk model \textcolor{blue}{(Simaan, 1993)}, we could then obtain the optimal portfolio. Empirical data show that the optimal portfolio obtained this way has less risk compared to an optimal portfolio of the classical BL model and that they become more negatively skewed as the expected returns of portfolios increase, 
which suggests that the investors trade a negative skewness for a higher expected return. We also observe a negative relation between portfolio volatility and portfolio skewness. This observation suggests that investors may be making a trade-off, opting for lower volatility in exchange for higher skewness, or vice versa. This trade-off indicates that stocks with significant price declines tend to exhibit increased volatility. 
\let\thefootnote\relax\footnotetext{We deeply would like to thank Stefan Nagel, Paul Zak, Barry Arnold, Yusif Simman, John Rutledge, Yaron Raviv, and Harold Vogel for their comments and suggestions to improve the reading of this paper.}
\newpage




\section{Introduction}
The Nobel Prize laureate Harry Markowitz \textcolor{blue}{(Markowitz, 1952)} provides a path-breaking mean-variance framework for portfolio optimization. This method, however, faces several problems such as highly-concentrated portfolios, input sensitivity problem, and estimation error maximization that are well-documented with asset pricing research when it is implemented in practice. Black and Litterman \textcolor{blue}{(Black $\&$ Litterman, 1990)} develop an alternative approach in portfolio optimization which allows investors to incorporate their expert view on how the market behaves in the future. The Black-Litterman (hereafter BL) model can help construct stable mean-variance efficient portfolios and overcome the input sensitivity problem. Both methods, however, assume that returns are normally distributed, an assumption consistent with other mainstream theories in finance.

As is well known now, asset returns exhibit non-normal behavior. In the presence of asset returns' non-normal behavior, optimal portfolio selection techniques should consider higher-moment risks. Skewness, a shape parameter, characterizes asymmetry of the returns of some financial assets, and is recognized early on in modern finance theory \textcolor{blue}{(Fama, 1963; Mandelbrot, 1963)}. For this reason, many authors propose different statistical distributions for the returns of different asset classes and focus on extending the model to be more general and suitable for many different types of portfolios \textcolor{blue}{(Bacmann $\&$ Benedetti, 2009; Hu $\&$ Kercheval, 2010; Giacometti, Bertocchi, Rachev $\&$ Fabozzi, 2007; Harvey, Liechty, Liechty $\&$ M$\ddot{\textnormal{u}}$ller, 2010; and Xiao $\&$ Valdez, 2015)}.

Recently, non-normal models of portfolio returns are documented in the financial literature because of their greater flexibility to accommodate skewness and larger fat tails. It is now feasible to employ better multivariate distribution families that capture skewness and heavy tails in the data. A number of past studies discuss the asymmetric characteristics of asset returns and attempt to extend the standard portfolio theory by considering asymmetry and fat tail to deal with non-normality in the distribution of asset returns and provide evidence that incorporating skewness into the portfolio decision causes major changes in the optimal portfolio.

\textcolor{blue}{Giacometti, Bertocchi, Rachev and Fabozzi (2007)} improve the classical BL model by applying more realistic models for asset returns (the normal, the t-student, and the stable distributions). 
By incorporating heavy-tailed distribution models for asset returns and alternative risk measures, they obtain the following results: (1) the appropriateness of the $\alpha$-stable distributional hypothesis is more evident when they compute the equilibrium returns and (2) the combination of $\alpha$-stable distribution and the choice of risk measure provides the best forecast.
\textcolor{blue}{Bacmann and Benedetti (2009)} propose an application of Bayesian allocation techniques to the portfolio selection problem in the hedge fund context. Given the strong departures from normality of the hedge fund returns, they suggest that the inclusion of higher moments and the parameter uncertainty should be addressed in the portfolio selection task. They implement and test a Bayesian framework for portfolio optimization process in order to consider different hedge fund styles as well as the estimation risk. 
They find that introducing skewness in the asset allocation task will produce a different allocation for investors with skewness preference.  \textcolor{blue}{Harvey, Liechty, Liechty and M$\ddot{\textnormal{u}}$ller (2010)} employ the skew normal distribution which has many attractive features for modeling multivariate returns. They address both higher moments and estimation risk in a coherent Bayesian framework. They start with the multivariate skew normal probability model developed by \textcolor{blue}{Sahu (2003)} and demonstrate that their generalization of the skew normal model of Sahu is able to capture higher moments. They show that skew normal models with a diagonal $\Delta$ (a diagonal matrix accommodating skewness) outperform the other models, with the Sahu model fitting best. For their empirical investigation, they find that including the skewness clearly improves the expected utility and the increasing utility indicates that substantially better portfolios are available for the data over a wide range of different skewness preferences. Recently, \textcolor{blue}{Xiao and Valdez (2015)} develop the explicit form of the posterior predictive distribution. They work with Meucci \textcolor{blue}{(Meucci, 2006a)}'s market-based version of the Black-Litterman model and make the extension to the elliptical market. Their resulting posterior provides solutions to optimization problems of asset allocation in the non-normal market assumption. 

Some studies suggest that investors may be willing to pay a premium for positive skewness in their portfolios \textcolor{blue}{(Friend $\&$ Westerfield, 1980)}. \textcolor{blue}{Leland (1999)} suggests that portfolio returns are not, in general, normally distributed and most investors have a preference for positively skewed returns, which implies that prices in equilibrium reflect more than mean and variance, and skewness will be positively valued by the market. 
\textcolor{blue}{Bekaert, Erb, Harvey and Viskanta (1998)}  show that there is significant skewness and kurtosis in the returns. 
They examine how asset allocation decisions are impacted in the presence of skewness and kurtosis. Their portfolio simulation shows that investment weights are increased toward the asset with positive skewness, holding kurtosis constant. 

Of particular interest to this paper is the relationship between mean return and skewness explored by  \textcolor{blue}{Kraus and Litzenberger (1976), Golec and Tamarkin (1998), Harvey and Siddique (2000), and Antti Ilman (2012)}. Interestingly, \textcolor{blue}{Golec and Tamarkin (1998)} study of horse race betting and argue that horse race bettors accept low-return, high-variance bets because they enjoy the high skewness offered by these bets. They suggest that bettors primarily trade off negative expected return and variance for positive skewness. 
\textcolor{blue}{Harvey and Siddique (2000)} suggest that if asset returns have systematic skewness, expected returns should include rewards for accepting this risk. In their study, 
they find that negative trade-off of mean return and skewness. That is, to get investors to hold low or negatively skewed portfolios, the expected return needs to be higher and vice versa. They conclude that systematic skewness is economically important and commands a risk premium. 
\textcolor{blue}{Antti Ilman (2012)} 
suggests that with the behavioral perspective, on the left tail, insurance-seeking investors focus on downside risk protection (limiting the left tail). On the right tail, lottery-seeking investors' positive skewness of returns implies the potential for greater variance of positive returns than negative returns. He concludes that 
buying insurance (limiting the left tail) and buying lottery tickets (enhancing the right tail) are popular activities and thus tend to be overpriced compared with actuarially neutral prices. Conversely, selling insurance and selling lottery tickets may boost long-term returns.

In this paper we introduce a three-moment BL model considering skewness risk using Simaan's model \textcolor{blue}{(Simaan, 1993)}.  When assets returns exhibit skewness, the classical mean-variance portfolio optimization is no longer valid. Simaan develops a framework for finding an optimal portfolio in this situation. 
The unique perspective of Simaan's model is to have two measures of risk. The model splits the variance of the portfolio return into two components called the spherical and the non-spherical variances. The returns are non-symmetric. In addition, the skewness is a function only of the non-spherical part of the variance. In his model, 
Simaan shows that minimizing the spherical variance of the portfolio return for given mean and non-spherical variance is the legitimate optimization method to select the most efficient portfolio.


Skew-normal models of portfolio returns are documented in the recent financial literature \textcolor{blue}{(Bacmann $\&$ Benedetti, 2009; Harvey, Liechty, Liechty $\&$ M$\ddot{\textnormal{u}}$ller, 2010; Carmichael $\&$ Co$\ddot{\textnormal{e}}$n, 2013; and Gan, 2014)} because of their greater flexibility to accommodate skewness and larger fat tails. The skew-normal distribution is introduced by Azzalini \textcolor{blue}{(Azzalini, 1985)} as a natural extension of the classical normal density to accommodate asymmetry. He extensively studies the properties of this distribution and extends this class to include the multivariate analog of the skew-normal. Arnold \textcolor{blue}{(Arnold, 1993)} introduces a more general skew-normal distribution as the marginal distribution of a truncated bivariate normal distribution in which $X$ is retained only if $Y$ satisfies certain constraints. Using this approach (hidden truncation model), more general univariate and multivariate skewed distributions are developed.

In this paper, we contribute to extend the BL model to 
the hidden truncation model (distribution) by departing from the normal market assumption. This class has the advantage of allowing for heavier tails and higher peaks than the ordinary normal distribution. By adopting the hidden truncation distribution, we capture the skewness of the asset returns and develop a more flexible model in a Bayesian framework. Our empirical results suggest that using the assumption of the skew-normal returns the hidden truncation BL model developed in this paper provides optimal portfolios with the same expected return but less risk compared to an optimal portfolio of the classical BL model. As $N$ (non-spherical variance) increases, a positive relationship between portfolio volatility and skewness is found at any level of $M$ (target portfolio return). We also find that there is a negative relationship between portfolio expected return and skewness at any level of $N$. In other words, the investors hold negatively (positively) skewed portfolios with higher (lower) expected returns. This supports the existing study that the investors trade a negative skewness for a higher expected return.

This paper proceeds as follows. Section 2 introduces a model for (hidden truncation) skewed normal returns and statistical properties of the multivariate hidden truncation distribution, studies Simaan model, and develops a portfolio optimization method under skew-normal assumption. Section 3 presents both classical and hidden truncation Black-Litterman model and the main theoretical frameworks underpinning this paper. Section 4 discusses empirical analysis. Section 5 concludes the paper. The appendices contain additional proofs, results, figures, and tables.


\section{A Model for Skewed Normal Returns}
In this section, we will briefly review a flexible family of skewed distributions developed by \textcolor{blue}{Arnold and Beaver (2000)} called the Hidden Truncation Skewed Normal Distribution. For interested readers, a short development of this distribution is provided in the Appendix. We then discuss portfolio optimization when returns follow this distribution.

\subsection{Multivariate Hidden Truncation Skew-Normal Distribution}
Following Arnold and Beaver, a $k-$dimensional random vector $X$ is said to have a hidden truncation skew normal distribution if its density is
$$
 f_{\boldsymbol{X}}(\boldsymbol{x}) = \frac{1}{{\Phi \left(\frac{\lambda_{0}}{\sqrt{1 + \boldsymbol{\lambda}_{1}^{t} \boldsymbol{\lambda}_{1}}}\right) }} \cdot \varphi(\boldsymbol{x}; \boldsymbol{\mu}, \Sigma) \cdot \Phi(\lambda_{0} + \boldsymbol{\lambda^{t}_{1}} \Sigma^{- 1/2} (\boldsymbol{x} - \boldsymbol{\mu})) \nonumber \label{SNT} \\
$$
In this case, we write $X\sim SNT_k(\lambda_0, \lambda_1, \mu, \Sigma)$. As an example, the density of a two dimensional case along with its contour are given below in Figure \ref{SNTplot12}. In this example, we assume that $\lambda_0 = 0, \lambda_1 = (1,2)^t, \mu = (0,0)^t$ and
$\Sigma = \begin{pmatrix}
           2 & 1\\
           1 & 2\\
          \end{pmatrix}.$

\begin{figure}%
	\centering
	\subfloat[Bivariate skew-normal density function]{{\includegraphics[width=10cm]{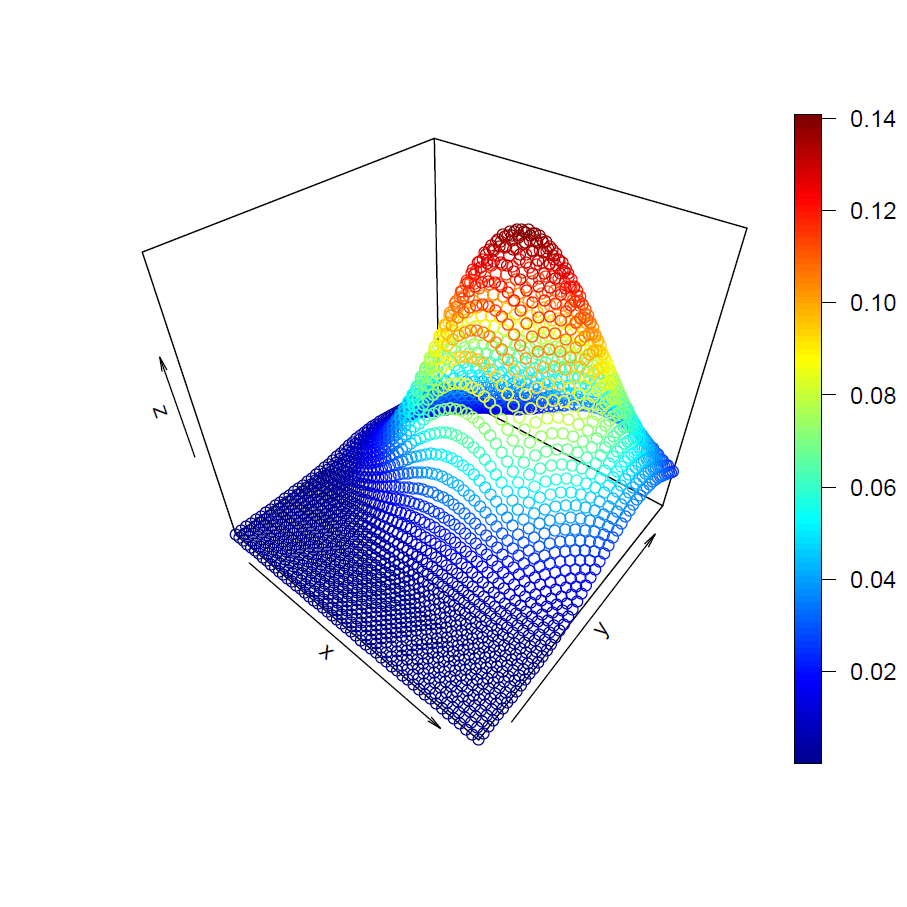}}}%
	\qquad
	\subfloat[Contour levels plot of bivariate skew-normal]{{\includegraphics[width=10cm]{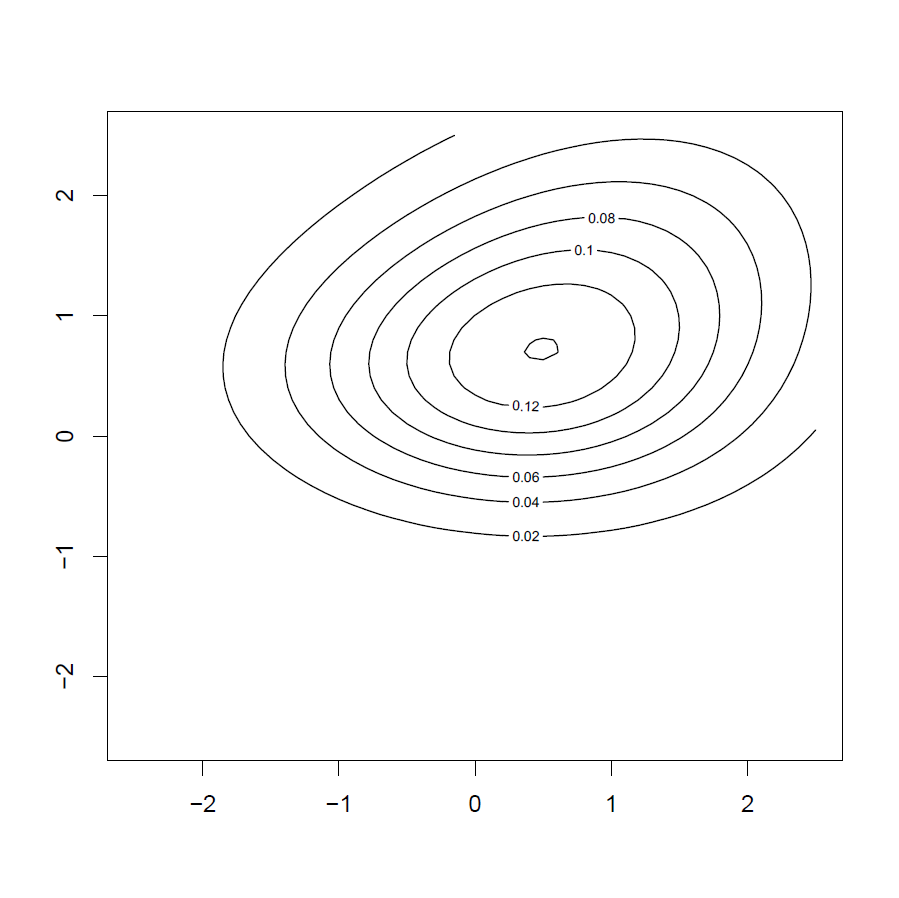}}}%
	\caption{Bivariate skew-normal density function and its contour from a SNT($\lambda_0=0, \lambda_1 = (1,2)^t, \mu = (0,0)^t, \Sigma$).}%
	\label{SNTplot12}%
\end{figure}


Among other important results, Arnold and Beaver show the following result.
\begin{thm}
	\textnormal{Suppose} $\boldsymbol{X} \sim SNT_{k}(\lambda_{0}, \boldsymbol{\lambda}_{1}, \boldsymbol{\mu}, \Sigma)$. Then its mgf is:
	\begin{align}
		M_{\bf X}({\bf s}) \:=\: \textnormal{exp} \left\lbrace {\bf s}^{t} \mu + \frac{1}{2} {\bf s}^{t} \Sigma {\bf s} \right\rbrace \cdot \frac{\Phi\left( \frac{\lambda_{0} + \boldsymbol{\lambda}^{t}_{1} \Sigma^{\frac{1}{2}} {\bf s}}{\sqrt{1 + \boldsymbol{\lambda}^{t}_{1} \boldsymbol{\lambda}_{1}}}\right) }{\Phi \left( \frac{\lambda_{0}}{\sqrt{1 + \boldsymbol{\lambda}^{t}_{1} \boldsymbol{\lambda}_{1}}}\right) }
	\end{align}
\end{thm}
\begin{proof}
See Appendix.
\end{proof}

From this result, they obtain the mean (and also the variance) of a hidden truncation skew-normal. In particular, if ${\bf X} \sim SNT_k(\lambda_0, \boldsymbol{\lambda}_1, \boldsymbol{\mu}, \Sigma)$, we have
\begin{align}
	\mathbb{E}({\bf X})
	&= \boldsymbol{\mu} + h(\lambda_0, \boldsymbol{\lambda}_1)
	\frac{\Sigma^{1/2}\boldsymbol{\lambda}_1}{\sqrt{1 + \boldsymbol{\lambda}_1^t\boldsymbol{\lambda}_1}} \label{eq:16}
\end{align}
where $h (\lambda_{0}, \boldsymbol{\lambda}_{1}) = \frac{\varphi \left( \frac{\lambda_{0}}{\sqrt{1 + \boldsymbol{\lambda}_{1}^{t} \boldsymbol{\lambda}_{1}}}\right) }{\Phi \left(\frac{\lambda_{0}}{\sqrt{1 + \boldsymbol{\lambda}_{1}^{t} \boldsymbol{\lambda}_{1}}}\right) }$.\\

We conclude this subsection with one very important result which has direct application to our portfolio optimization. The proof of this theorem is in Appendix.

\begin{thm} \textnormal{\textbf{Affine Transformation}}.
\label{Linear Transformation}
	Suppose $\boldsymbol{X} \sim SNT_{n}(\lambda_{0}, \boldsymbol{\lambda}_{1}, \boldsymbol{\mu}, \Sigma)$. Let $\boldsymbol{Y} =  A\boldsymbol{X} + {\bf b}$ where $A$ is an $m \times n$ matrix and ${\bf b} \in \mathbb{R}^m.$ Then ${\bf Y} \sim SNT_{m}(\tau_{0}, \boldsymbol{\tau}_{1}, \boldsymbol{\mu}_{y}, \Sigma_{y})$, where
	\begin{align*}
		\boldsymbol{\mu}_{y} \:&=\: {\bf b} + A \boldsymbol{\mu} \\
		\Sigma_{y} \:&=\: A \Sigma A^{t} \\
		\tau_{0} \:&=\: \frac{\lambda_{0}}{\sqrt{1 + \boldsymbol{\lambda}^{t}_{1} \left[ I - H \Sigma_{y}^{-1} H^{t} \right] \boldsymbol{\lambda}_{1}}} \\
		\boldsymbol{\tau}_{1} \:&=\: \frac{\Sigma_{y}^{-\frac{1}{2}}H^t \boldsymbol{\lambda}_{1}}{\sqrt{1 + \boldsymbol{\lambda}^{t}_{1} \left[ I - H \Sigma_{y}^{-1} H^{t} \right] \boldsymbol{\lambda}_{1}}} \\
		H \:&=\: \Sigma^{\frac{1}{2}} A^{t}
	\end{align*}
\end{thm}

\begin{proof}
	See Appendix.
\end{proof}

\begin{cor}
\label{PortSNT}
	\textnormal{Assume} $\boldsymbol{X}$ $\sim SNT_{k}(\lambda_{0}, \boldsymbol{\lambda}_{1}, \boldsymbol{\mu}, \Sigma, )$ \textnormal{and} $A_{m \times n}$. For any ${\boldsymbol{w}}\in \mathbb{R}^k$, let $Y = {\boldsymbol{w}}^{t} \boldsymbol{X}$. Then $Y \sim SNT_{1}(\tau_{0}, \tau_{1}, \mu_y, \Sigma_y)$, where
	\begin{align*}
		\mu_y \:\:=&\:\:\: {\boldsymbol{w}}^{t}\boldsymbol{\mu}\\
		\Sigma_{y} \:\:=&\:\:\: {\boldsymbol{w}}^{t} \Sigma {\boldsymbol{w}} \\
		\tau_{0} \:\:=&\:\:\: \frac{\lambda_{0}}{\sqrt{1 + \boldsymbol{\lambda}^{t}_{1} \left[ I - H (\Sigma_{y}^{-1}) H^{t} \right] \boldsymbol{\lambda}_{1}}} \\
		\tau_{1}
		=&\:\:\: \frac{(\Sigma_{y})^{-\frac{1}{2}} H^{t} \boldsymbol{\lambda}_{1}}{\sqrt{1 + \boldsymbol{\lambda}^{t}_{1} \left[ I - H (\Sigma_{y}^{-1}) H^{t} \right] \boldsymbol{\lambda}_{1}}} \\
		H \:\:=&\:\:\: \Sigma^{\frac{1}{2}} {\boldsymbol{w}}
	\end{align*}
\end{cor}

In the next subsection, we discuss Simaan’s model and portfolio optimization algorithm under hidden truncation skew-normal returns.

\subsection{Simaan's Model and Portfolio Optimization under Skew-Normal Assumption}
When asset returns exhibit skewness, the classical mean-variance on portfolio optimization is no longer valid. \textcolor{blue}{Simaan (1993)} develops a theoretical framework for finding an optimal portfolio in this situation.

\begin{defn}
	Let ${\bf \mu}$ and ${\bf b}$ be two constant vectors in $\mathbb{R}^k$. Simaan's model assumes that assets returns ${\bf R}$ has the following representation
	$$
	{\bf R} = {\bf \mu} + {\bf b}Y + {\bf \epsilon}
	$$
	where $Y$ is a univariate non-elliptical distribution and ${\bf \epsilon}$ is a random vector whose distribution conditionally on $Y$ is a multivariate spherical distribution with a characteristic matrix $W$.
\end{defn}
With this representation, Simaan shows that, among other results, $\mathbb{E}({\bf R}) = {\boldsymbol \mu} + b\mathbb{E}(Y),
V({\bf R}) = kW + {\bf b}{\bf b}^t\sigma_y^2$ for some positive constant $k$. The two components of the variance are called the spherical (associated with $\epsilon$) and non-spherical (associated with $Y$) components of the variance. Thus for a given portfolio $\boldsymbol{w}$, the variance of the portfolio return, $R_w = \omega^t{\bf R}$ is
\begin{align*}
	V(R_{w}) &= V(\boldsymbol{w}^{t}{\bf R})
	= \boldsymbol{w}^{t} kW \boldsymbol{w} + \sigma_y^2 \boldsymbol{w}^{t}{\bf b}{\bf b}^t\boldsymbol{w}.
\end{align*}

It is then natural to see if there is any relationship between Simaan's model and the hidden truncation skewed-normal assumption. The following theorem, which is shown by \textcolor{blue}{Azzalini and Valle (1996)}, shows the connection between the two. (See Appendix for proof).

\begin{thm}
\label{SimaanModel}
	Suppose $X \sim N(0,1)$ is independent with ${\bf Z} \sim N_k({\bf 0}, \Psi)$ where $Z_i \sim N(0,1), i=1,2,\ldots,k.$ (So we can view $\Psi$ as a correlation matrix.) Let $D = \text{diag}(\delta_1, \delta_2, \ldots, \delta_k)$ with $\delta_j \in (-1,1)$ and $W = \text{diag}(w_1, w_2, \ldots, w_k)$ with $w_j > 0.$ Then, for ${\boldsymbol \mu} \in \mathbb{R}^k$ and ${\boldsymbol \delta} = (\delta_1, \ldots, \delta_k)^t$,
	\begin{align}
	{\bf R} = {\boldsymbol \mu} + W{\boldsymbol \delta}|X| + W(I - D^2)^{1/2}{\bf Z}
	\end{align}
	follows the hidden truncation skew normally distributed $SNT_k(\lambda_0, \boldsymbol{\lambda}_{1}, {\bf \mu}, \Sigma),$ where
	\begin{align*}
	\bar{\Sigma} &= {\boldsymbol \delta \boldsymbol \delta}^t + (I - D^2)^{1/2}\Psi(I - D^2)^{1/2}\\
	{\boldsymbol \alpha} &= \frac{\bar{\Sigma}^{-1}{\boldsymbol \delta}}
	{(1 - {\boldsymbol \delta}^t\bar{\Sigma}^{-1}{\boldsymbol \delta})^{1/2}}\\
	\Sigma &= W\bar{\Sigma}W\\
    \lambda_0 &= 0\\
    \boldsymbol{\lambda}_{1} &= \Sigma^{1/2}W^{-1}{\boldsymbol \alpha}.
	\end{align*}
\end{thm}

\begin{proof}
See Appendix.
\end{proof}

In practice, we are concerned with the other direction. That is, given ${\bf R} \sim SNT_k(\lambda_0 = 0, {\boldsymbol \lambda_1}, {\boldsymbol \mu}, \Sigma),$ we need to determine ${\boldsymbol \delta}, W$ and $\Psi$ so that the following identity
$$
{\bf R} = \boldsymbol \mu + W{\boldsymbol \delta}|X| + W(I - D^2)^{1/2}{\bf Z}
$$
holds. It can be shown that, after some algebras, if we take $W$ to be the diagonal matrix of the standard deviations of $\Sigma$, $\bar{\Sigma}$ be the correlation matrix of $\Sigma$, $\delta = \frac{\overline{\Sigma}\alpha}{(1+\alpha^t\overline{\Sigma}\alpha)^{1/2}},$ with $\boldsymbol \alpha = W\Sigma^{-1/2}\boldsymbol \lambda_1, D$ be the diagonal matrix of the $\delta$ vector, and
$$
\Psi = (I - D^2)^{-1/2}(\bar{\Sigma}-\delta \delta^t)(I-D^2)^{-1/2}
$$
then indeed, ${\bf R}$ has the above representation. With this representation, we can find the variance of asset returns as follows:
\begin{align}
V &= Var({\bf R}) = Var({\boldsymbol \mu} + W{\boldsymbol \delta}|X| + W(I - D^2)^{1/2}{\bf Z}) \nonumber \\
&= \sigma^2_{|X|}(W{\boldsymbol \delta})(W{\boldsymbol \delta})^t +
W(I - D^2)^{1/2}\Psi (I-D^2)^{1/2}W \nonumber \\
&= \sigma^2_{|X|}{\bf b}{\bf b}^t + S,
\end{align}
where ${\bf b} = W{\boldsymbol \delta}$ and $S = W(I - D^2)^{1/2}\Psi (I-D^2)^{1/2}W = W(\bar{\Sigma} - \delta\delta^t)W = \Sigma - {\boldsymbol b}{\boldsymbol b}^t.$ So the variance is sum of the non-spherical variance (associated ${|X|}$) and \text{spherical variance} (associated ${\bf Z}$) components as termed by Simaan. As a result, the portfolio variance is given by
\begin{align}
\label{portVar}
V_w = Var(R_w) = \sigma^2_{|X|}b_w^2 + {\boldsymbol{w}}^{t}S{\boldsymbol{w}},
\end{align}
where $b_w = {\boldsymbol{w}}^{t}{\bf b}.$ Using Simaan's observation that the utility of $R_w$ is a non-increasing function of the portfolio spherical component of the variance while holding the non-spherical component fixed, the optimal portfolio can be found by minimizing the spherical variance while holding the non-spherical variance fixed. In other words, the optimal portfolio is obtained by solving the following optimization problem:
\begin{align} 
\underset{\boldsymbol{w}}{\textnormal{min}} \:\:\:\:\:\:\:\:\:\:\:\:\:\:\:\:\: f({\boldsymbol{w}}) \:&=\: \frac{1}{2}\boldsymbol{w}^{t}S{\boldsymbol{w}} \\~\nonumber\\
\textnormal {subject to} \:\:\:\:\:\:\:\:\: g_1({\boldsymbol{w}}) \:&=\ \boldsymbol{w}^{t}{\bf 1} \:=\: 1 \nonumber\\
g_2({\boldsymbol{w}}) \:&=\ \boldsymbol{w}^{t}\mu \:=\: \mu_w \:=\: M \nonumber \\
g_3({\boldsymbol{w}}) \:&=\ \boldsymbol{w}^{t}{\bf b} \:=\: b_w \:=\: N. \nonumber
\end{align}

By applying the Lagrange multiplier method (See Appendix), the optimal portfolio $\boldsymbol{w}$ can be put in the form $c_1\boldsymbol{w}_{1} + c_2\boldsymbol{w}_{2} + c_3\boldsymbol{w}_{3}$, where
	\begin{align}
	\boldsymbol{w}_{1} = \frac{S^{-1}\mathbf{1}}{\mathbf{1}^tS^{-1}\mathbf{1}},\:\: \boldsymbol{w}_{2} = \frac{S^{-1}\mu}{\mathbf{1}^tS^{-1}\mu} \:\:\: \textnormal{and} \:\:\: \boldsymbol{w}_{3} = \frac{S^{-1}{\bf b}}{\mathbf{1}^t S^{-1}{\bf b}}.
	\end{align}
The weights $c_1, c_2, c_3$ can be determined once the target portfolio mean $M$ and the target non-spherical variance $N$ are given.

\section{Black-Litterman Models}
\textcolor{blue}{Black and Litterman (1991 and 1992)} propose a one-period mean-variance optimization (MVO) in which the expected risk premia  of the assets incorporate views formulated by securities and market analysts. This model can be viewed as the Bayesian updating of a prior distribution of risk premia into a posterior distribution reflecting the views. 
In this section, we briefly review the classical Black-Litterman model and we then show how to extend it under hidden truncated skew-normal distribution, which is the main result of this article.

\subsection{Classical Black-Litterman (BL) Model}
Let ${\bf R}$ denote the asset returns. In this popular model, Black and Litterman put the following prior distribution on the vector of expected returns
$$
{\bf M} \sim N_k(\boldsymbol{\pi}, \tau \Sigma)
$$
and assume that ${\bf R} | {\bf M} = {\bf m} \sim N_k({\bf m}, \Sigma)$ for some positive definite matrix $\Sigma.$ In practice, $\Sigma$ is taken to be the covariance matrix of the data and $\boldsymbol{\pi}$ is the implied excess returns. The implied excess returns are the set of returns that clear the market. The vector of implied excess returns are derived as follows: 
\begin{align}
	\boldsymbol{\pi} \:&=\: \gamma \Sigma \boldsymbol{w}^{*}_{}, \label{EquiReturns}
\end{align}
where
\begin{align}
	\boldsymbol{\pi} \:\:&\:\: \textnormal{vector of implied excess returns} \nonumber \\
	\gamma \:\:&\:\ \textnormal{risk aversion coefficient} \nonumber \\
	\Sigma \:\:&\:\ \textnormal{covariance matrix of excess returns} \nonumber \\
	\boldsymbol{w}^{*}_{} \:\:&\:\ \textnormal{market capitalization weights of the assets} \nonumber
\end{align}

The risk aversion coefficient, $\gamma$, characterizes the expected risk-return tradeoff and set as 2.5 \textcolor{blue}{(He and Litterman, 1999)}. That is, it is the rate at which an investor will trade expected return for less risk. In the reverse optimization process, $\gamma$ acts as a scaling factor; more excess return per unit of risk (a larger $\gamma$) increases the estimated excess return.

In addition, the BL framework allows investors to inject their own opinions through a view vector $V \in \mathbb{R}^d, d\le k$ and a pick $d\times k$ matrix $P.$ In particular, the model assumes that
$$
{\bf V} | {\bf M = m} \sim N_d(P{\bf m}, \Omega_v).
$$
where $\Omega_v$ is a $d\times d$ diagonal matrix representing how confidence the investors are on their views. This distribution serves as the likelihood function. Together with the prior distribution of ${\bf M}$, we obtain the following posterior distribution:
\begin{align*}
	\pi^{post}_{{\bf M|V=v}}({\bf m}) &\propto \pi^{pr}_{\bf M}({\bf m})
	f_{V|{\bf M = m}}({\bf v})\\
	&\propto \varphi(\boldsymbol{m};\boldsymbol{\pi}, \tau \Sigma)\varphi({\bf v}; P{\bf m}, \Omega_v)\\
	&\propto \varphi({\bf m}; \boldsymbol{\mu}_{BL}, \Sigma_{BL})h({\bf v}),
\end{align*}
where
\begin{align*}
	\boldsymbol{\mu}_{BL} &= \Sigma_{BL}\left[(\tau \Sigma)^{-1}\boldsymbol{\pi} + P^t\Omega_v^{-1}{\bf v}\right]\\
	\Sigma_{BL} &= \left[(\tau \Sigma)^{-1} + P^t\Omega_v^{-1}P\right]^{-1}
\end{align*}
and $h({\bf v})$ is independent of ${\bf m}.$ (See Theorem \:\ref{ProdGaussian1}\: in Appendix.) Thus the posterior distribution is ${\bf M} | {\bf V = v} \sim N_k(\boldsymbol{\mu}_{BL}, \Sigma_{BL}).$ The ``predictive distribution" ${\bf R} | {\bf V = v}$ then can be obtained as follows.
\begin{align*}
	f_{{\bf R} | {\bf V = v}}({\bf r}) = \int f_{{\bf R}|{\bf M=m}}({\bf r})\pi^{post}_{{\bf M}|{\bf V=v}}({\bf m})d{\bf m}.
\end{align*}
Using Theorem \:\ref{ProdGaussian2}\: in Appendix, the integrand of the above integral can be written as
\begin{align*}
	f_{{\bf R}|{\bf M=m}}&({\bf r})\pi^{post}_{{\bf M}|{\bf V=v}}({\bf m})\\
	&= \varphi({\bf r};{\bf m}, \Sigma)\varphi({\bf m};\boldsymbol{\mu}_{BL}, \Sigma_{BL})\\
	&= \varphi({\bf r}; \boldsymbol{\mu}_{BL}, \Sigma_{BL}+\Sigma)
	\varphi({\bf m}; z({\bf r}, \boldsymbol{\mu}_{BL}), \Delta),
\end{align*}
where $\Delta = (\Sigma^{-1}+\Sigma_{BL}^{-1})^{-1}$ and $z({\bf r}, \boldsymbol{\mu}_{BL}) = \Delta(\Sigma^{-1}{\bf r} + \Sigma^{-1}_{BL}\boldsymbol{\mu}_{BL}).$ Thus by integrating out ${\bf m}$, we get the following predictive distribution ${\bf R} | {\bf V =v} \sim N_k(\boldsymbol{\mu}_{BL}, \Sigma_{BL} + \Sigma).$ So given a portfolio ${\bf w}$, the ``predictive distribution" of the asset returns follows
$$
{\bf R}_w = {\boldsymbol w}^{t}{\bf R} \sim N(\mu_{BL}^w, \Sigma_{BL}^w)
$$
where $\mu_{BL}^w = \boldsymbol{w}^{t}{\bf \mu}_{BL}$ and $\Sigma_{BL}^w = \boldsymbol{w}^{t}(\Sigma_{BL}+\Sigma)\boldsymbol{w}.$ The term ``predictive distribution" is a bit misleading in this case. In Bayesian statistics, this term refers to the distribution of a new observation given the observed data. Here it refers to the distribution of the returns given views. It's more like a posterior distribution. However, the process in obtaining this distribution is similar to that of a typical predictive distribution.

The normality of the ``predictive distribution" allows us to find the optimal portfolio by using the mean-variance method. In other words, by setting $\mu_{BL}^w = M \in \mathbb{R}$ we can write the appropriate quadratic programming as follows:
\begin{align}
	{\textnormal{min}} \:&\:\: \Sigma^{w}_{BL} \nonumber \\
	\text{ Subject to } \mu^{w}_{BL} &\: = \: M \nonumber \\
	\mathbf{1}^{t} \boldsymbol{w} &\: = \: 1.
	\label{MV-opt}
\end{align}

This optimization can be solved using Markowitz's mean-variance framework.

\subsection{Hidden Truncation Skew-Normal Black-Litterman Model}
We now assume that asset returns are hidden truncation skew-normal. The main goal of this subsection, which is also the main goal of this article, is to obtain the ``predictive distribution" of the asset returns once investor views are incorporated. We then apply this result to find the optimal portfolio. Recall that if ${\bf X} \sim SNT_k(\lambda_0, \boldsymbol{\lambda}_1, \boldsymbol{\mu}, \Sigma)$, then
\begin{align}
	\mathbb{E}({\bf X})
	&= \boldsymbol{\mu} + h(\lambda_0, \boldsymbol{\lambda}_1)
	\frac{\Sigma^{1/2}\boldsymbol{\lambda}_1}{\sqrt{1 + \boldsymbol{\lambda}_1^t\boldsymbol{\lambda}_1}} \label{eq:16},
\end{align}
where $h (\lambda_{0}, \boldsymbol{\lambda}_{1}) = \frac{\varphi \left( \frac{\lambda_{0}}{\sqrt{1 + \boldsymbol{\lambda}_{1}^{t} \boldsymbol{\lambda}_{1}}}\right) }{\Phi \left(\frac{\lambda_{0}}{\sqrt{1 + \boldsymbol{\lambda}_{1}^{t} \boldsymbol{\lambda}_{1}}}\right) }$. For compactness, let ${\bf s} = h(\lambda_0, \boldsymbol{\lambda}_1)\frac{\Sigma^{1/2}\boldsymbol{\lambda}_1}{\sqrt{1 + \boldsymbol{\lambda}_1^t\boldsymbol{\lambda}_1}}.$ Then we get $E({\bf X}) = \boldsymbol{\mu} + {\bf s}.$\\

In this subsection, similar in the classical BL model, we continue to assume that
\begin{align*}
	{\bf M} &\sim N_k(\boldsymbol{\pi}, \tau \Sigma)\\
	{\bf V} | {\bf M = m} &\sim N_d(P{\bf m}, \Omega_v)
\end{align*}
However, as for asset returns, we assume that
$$
{\bf R}|{\bf M =m} \sim SNT_k(\lambda_0, \boldsymbol{\lambda}_1, {\bf m - s}, \Sigma).
$$

As a result, the posterior distribution ${\bf M} | {\bf V = v} \sim N_k(\boldsymbol{\mu}_{BL}, \Sigma_{BL})$ as shown in the previous section. It remains to determine the ``predictive distribution" of ${\bf R}|{\bf V = v}.$ We now state the main result of this article.

\begin{thm}
\label{mainResult}
Assume that ${\bf M} \sim N_k(\boldsymbol{\pi}, \tau \Sigma), {\bf V} | {\bf M = m} \sim N_d(P{\bf m}, \Omega_v)$ and that ${\bf R}|{\bf M =m} \sim SNT_k(\lambda_0, \boldsymbol{\lambda}_1, {\bf m - s}, \Sigma).
$ Then
\begin{align}
	{\bf R}|{\bf V = v} \sim SNT_k(\tau_0, \boldsymbol{\tau}_1, \boldsymbol{\mu}_{BL}-{\bf s}, \Sigma + \Sigma_{BL}), \label{SNReturns}
\end{align}
where $\tau_0 = \frac{\lambda_0}{\sqrt{1+(\boldsymbol{\lambda}_1^*)^t\boldsymbol{\lambda}_1^*}}$
and $\boldsymbol{\tau}_1 = \frac{(\Sigma+\Sigma_{BL})^{-1/2}\Sigma^{1/2}\boldsymbol{\lambda}_1}
{\sqrt{1+(\boldsymbol{\lambda}_1^*)^t\boldsymbol{\lambda}_1^*}},$ and $\boldsymbol{\lambda}_1^* = - \Delta^{1/2}\Sigma^{-1/2}\boldsymbol{\lambda}_1.$
\end{thm}

\begin{proof}
See Appendix.
\end{proof}

Thus for a given portfolio $\boldsymbol{w}$, the distribution of the portfolio return $R_w = \boldsymbol{w}^{t}{\bf R}$ can be readily shown to follow a skew normal distribution. Thus the classical mean-variance portfolio optimization is no longer applicable. In this case, to find the optimal portfolio, we use Simaan's model. As shown in the earlier section, this is equivalent of assuming that ${\bf R|M=m}$ follows a hidden truncated skew normal distribution as in the above theorem with $\lambda_0 = 0$. 

For the special case $\lambda_0 = 0,$ hence, $\tau_0 = 0,$ we have
\begin{align}
	{\bf R}|{\bf V = v} \sim SNT_k(0, \boldsymbol{\tau}_1, \boldsymbol{\mu}_{skBL}, \Sigma_{skBL}) \label{eq:18},
\end{align}
where $\boldsymbol{\mu}_{skBL} = \boldsymbol{\mu}_{BL} - \textbf{s}, \Sigma_{skBL} = \Sigma+\Sigma_{BL}$ and ${\bf s} = \sqrt{\frac{2}{\pi}}\frac{{\bf \Sigma}^{1/2}\boldsymbol{\lambda}_1}{\sqrt{1 + \boldsymbol{\lambda}_1^t\boldsymbol{\lambda}_1}}.$ Since $R_{\boldsymbol{w}} | {\bf V=v}$ is a linear transformation of ${\bf R} | {\bf V =v},$ based on Corollary (\ref{PortSNT}), we thus have
\begin{align}
	R_{\boldsymbol{w}} | {\bf V=v} \:\sim\: SNT(0, \tau^{w}_{1},  \mu^{w}_{skBL}, \Sigma^{w}_{skBL}), \label{eq:19}
\end{align}
where
\begin{align*}
	\mu^{w}_{skBL} \:&=\: \boldsymbol{w}^{t}\mu_{skBL} \\
	\Sigma^{w}_{skBL} \:&=\: \boldsymbol{w}^{t}(\Sigma_{skBL})\boldsymbol{w} \\
	\tau^{w}_{1} \:&=\: \frac{(\Sigma^{w}_{skBL})^{-\frac{1}{2}} H^t \boldsymbol{\tau}_{1}}{\sqrt{1 + \boldsymbol{\tau}^{t}_{1} \left[ I - H\: (\Sigma^{w}_{skBL})^{-1} H^{t} \right] \boldsymbol{\tau}_{1}}} \\
	H \:&=\: (\Sigma_{skBL})^{\frac{1}{2}} \boldsymbol{w}
\end{align*}

In this representation, we are ready to find the optimal portfolio by applying the technique discussed in Section 2. In particular, let $V_{sp}^w = \boldsymbol{w}^{t}(\Sigma_{skBL} - \mathbf{bb}^{t})\boldsymbol{w}$ denote the spherical component of the variance where ${\bf b} = W{\boldsymbol \delta}$ where $W$ is the diagonal matrix of the standard deviations of $\Sigma_{skBL}, \delta = \frac{\overline{\Sigma}_{skBL}\alpha}{(1+\alpha^t\overline{\Sigma_{skBL}}\alpha)^{1/2}},$ with $\boldsymbol \alpha = W\Sigma_{skBL}^{-1/2}\boldsymbol \tau_1$ as discussed in the earlier section. Then the optimal portfolio can be obtained by solving the following problem:
\begin{align}
	\underset{w}{\textnormal{min}} \:\:\:\:\:\:\:\:\:\:\:\:\: \boldsymbol{w}^{t}(\Sigma_{skBL} - \mathbf{b} \mathbf{b}^{t}) \boldsymbol{w} \label{eq:23} \\
	\textnormal {subject to} \:\:\:\:\: \boldsymbol{w}^{t}\mu_{skBL} \:&=\: M \nonumber \\
	\boldsymbol{w}^{t}b \:&=\: N \nonumber\\
	\boldsymbol{w}^{t}\mathbf{1} \:&=\: 1. \nonumber
\end{align}

Once the optimal portfolio has been found, at a given $M$ and $N$, the variance of this optimal portfolio can be obtained based on the result of Equation \ref{portVar}. In particular, with $b_w = \boldsymbol{w}^{t}b,$ we have
\begin{align}
	V(R_w) &= V_{np} + V_{sp} = \sigma^2_{|X|}b^2_w + \boldsymbol{w}^{t}(\Sigma_{skBL} - \mathbf{bb}^{t})\boldsymbol{w} \nonumber \\
	&= (1 - \frac{2}{\pi})b^2_w + \Sigma_{skBL}^w - b^2_w \nonumber \\
	&= \Sigma_{skBL}^w - \frac{2}{\pi}N^2. \label{eq:24}
\end{align}

\newpage

\section{Empirical Analysis}
In this section, we outline and demonstrate the Black-Litterman allocation procedure in skew-normal markets, which involves several major steps. We first analyze the time series data to estimate the key parameters of the hidden truncation skew-normal distribution. Then after incorporating views, the optimal portfolios are obtained for various values of the target means and non-spherical variances. For comparisons, the classical Black-Litterman optimal portfolios are also obtained.

\subsection{Data}
For illustration, we consider in this paper 217 observations of monthly log-returns from December 2004 to November 2022 on 13 
equities: Apple, Amazon, Google, Microsoft, ExxonMobil, Intel, Verizon, Coca-Cola, Netflix, Comcast, NVIDIA, Starbucks, and Walmart. Figure~\ref{TS} presents the monthly stock prices for the 13 equities. 
Figures~\ref{Histograms} and~\ref{QQ} present a histogram and Q--Q plot for the 13 equities. Note that they indicate the return distributions are skewed.

\begin{landscape}
	\begin{figure}[h!]
		\centering
		{{\includegraphics[width=22cm]{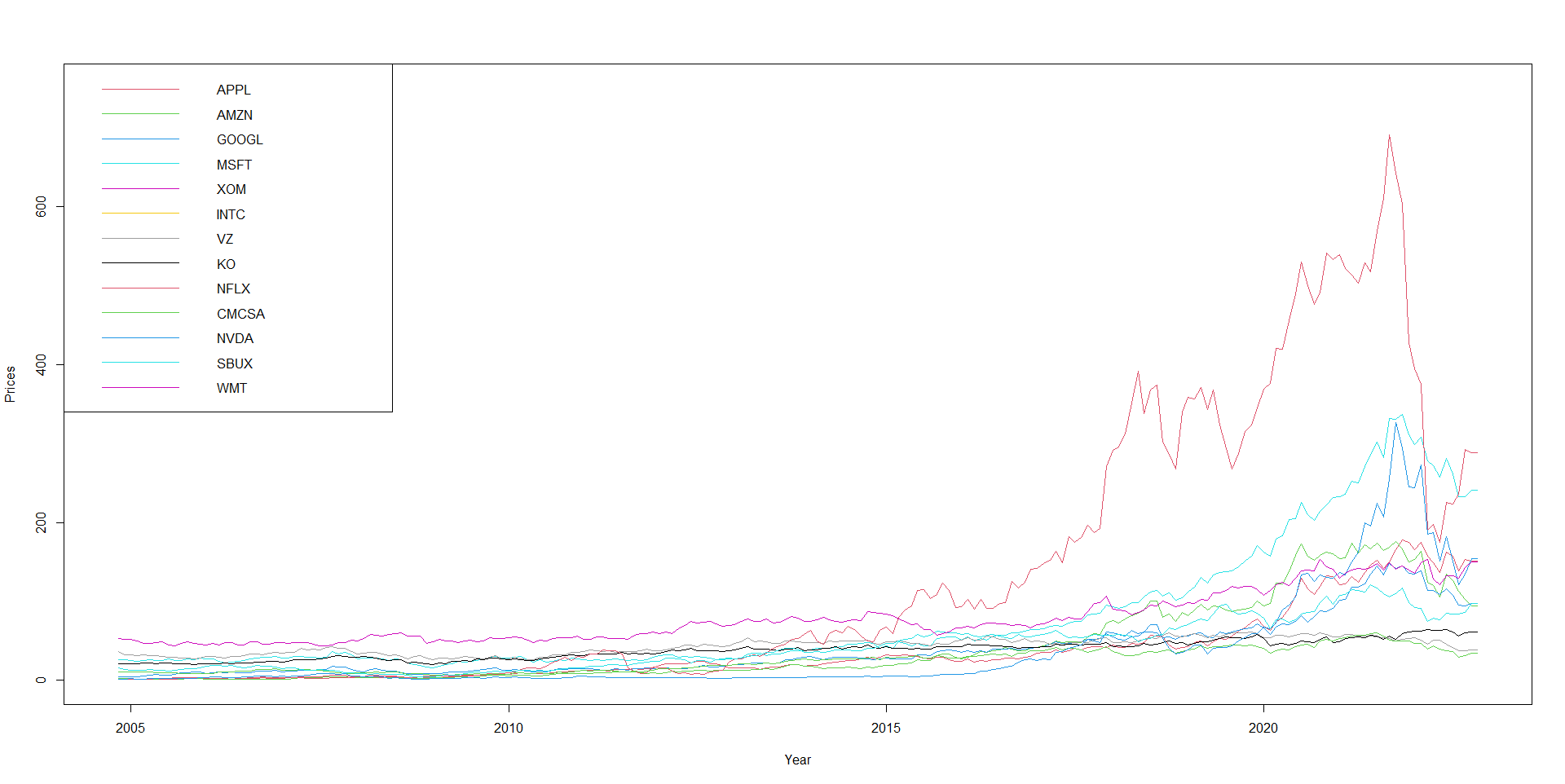}}}%
		\qquad
		\caption{\small Monthly Stock Prices. Figure 4 presents monthly stock prices for the 13 equities for the period from December 2004 to November 2022.}%
		\label{TS}%
	\end{figure}
\end{landscape}

\begin{landscape}
	\begin{figure}[h!]%
		\centering
		{{\includegraphics[width=22cm]{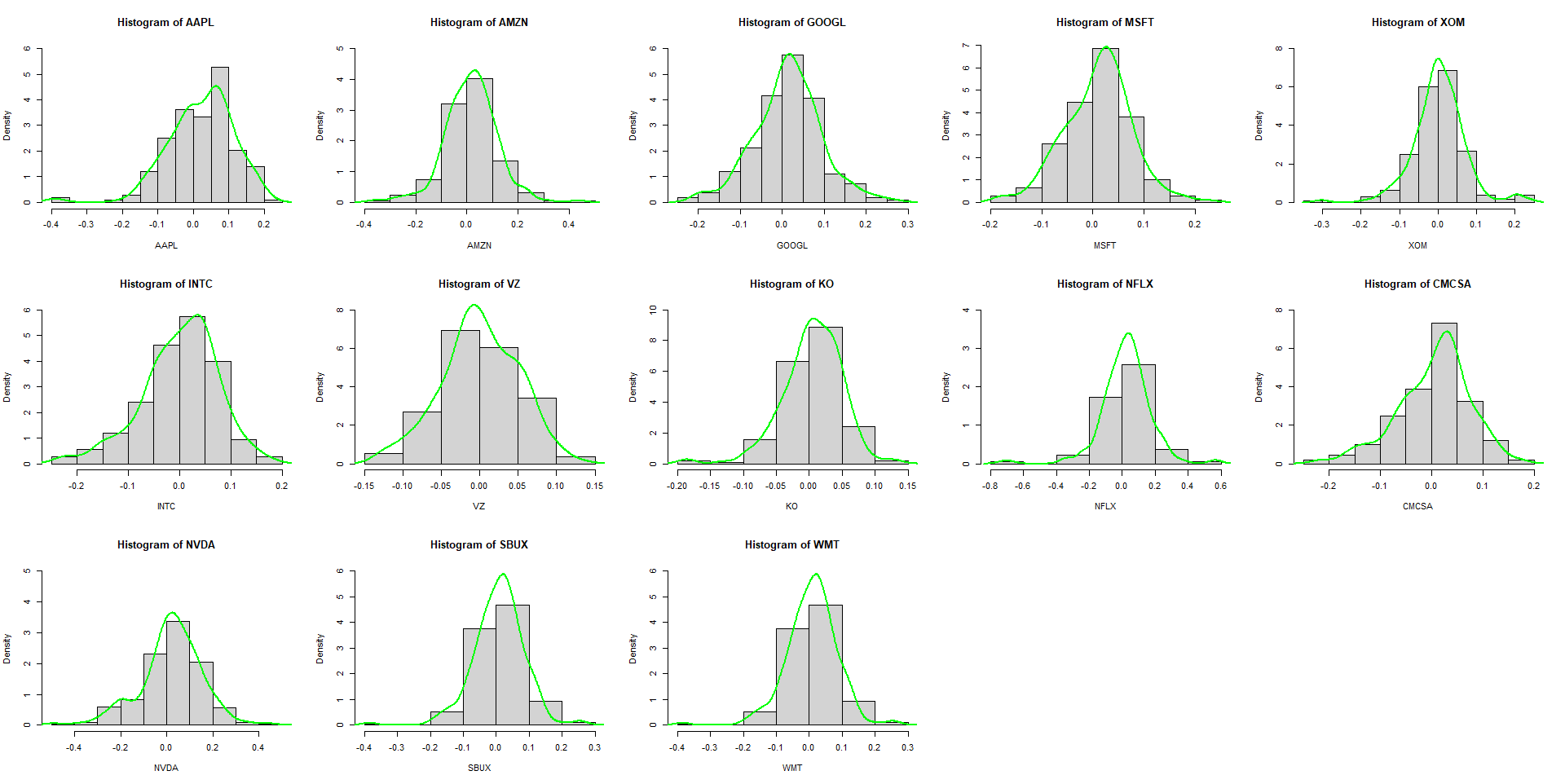}}}%
		\qquad
		\caption{\small Log Returns Distribution. Figure 3 presents the log returns distribution of the 13 equities for the period from December 2004 to November 2022.}%
		\label{Histograms}%
	\end{figure}
\end{landscape}

\begin{figure}[h!]
	\centering
	{{\includegraphics[width=13cm]{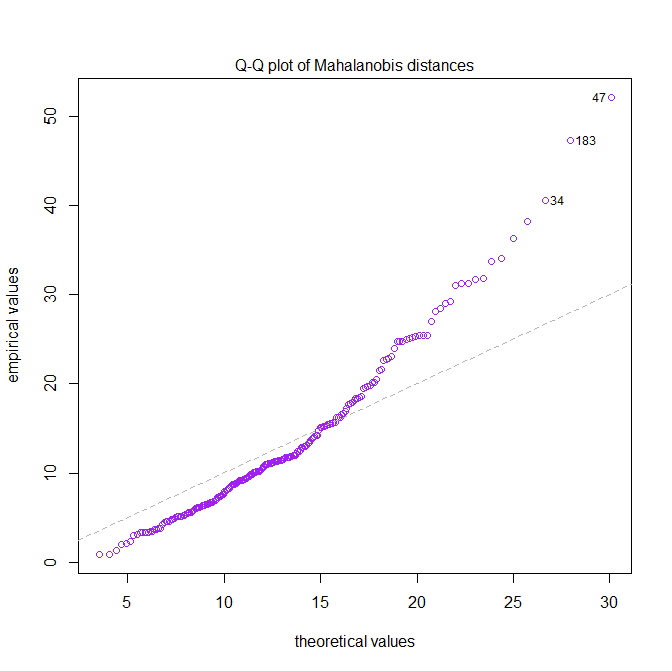}}}%
	\qquad
	\caption{\small Quantile-quantile (Q--Q) plot. Figure 4 presents a Q-Q plot for the 13 equities. 
		As shown, the points forming curves deviate markedly from a straight line.}%
	\label{QQ}%
\end{figure}

\subsection{Black-Litterman Allocation in Normal Markets}
Now the investor has views on the 13 stocks in her portfolio. The market models are give by:
\begin{align}
\boldsymbol{R} | \boldsymbol{M = m} \:&\sim\: N_{13}(\boldsymbol{m}, \Sigma)\\
\boldsymbol{M} \:&\sim\: N_{13}(\boldsymbol{\pi}, \tau \Sigma),
\end{align}
where $\boldsymbol{\pi}$ is the vector of implied excess
returns (see Equation \ref{EquiReturns}) and $\tau = 0.025$ \textcolor{blue}{(He and Litterman, 1999)}. Table~\ref{BL-largecap-summary} reports the vector of the market capitalization weights $\boldsymbol{w}^{*}_{}$, historical return $\boldsymbol{\hat{\mu}}_{}$, implied excess return $\boldsymbol{\pi}_{}$, and Black-Litterman mean vector $\boldsymbol{m}_{BL}$. Table~\ref{BL-largecap-covar} shows the covariance matrix $\widehat{\Sigma_{}}$.\\

\begin{table}[!htbp] \small
\caption{\small This table reports the vector of the market capitalization weights $\boldsymbol{w}^{*}_{}$, historical return $\boldsymbol{\hat{\mu}}_{}$, implied excess return $\boldsymbol{\pi}_{}$, and Black-Litterman mean vector $\boldsymbol{m}_{BL}$.} 
\centering 
\begin{tabular}{c c c c c c } 
\hline 
Ticker & $\hat{\boldsymbol{w}}^{*}_{}$ &  $\hat{\boldsymbol{\mu}}_{}$ & $\hat{\boldsymbol{\pi}}_{}$ & $\hat{\boldsymbol{m}}_{BL}$ \\ [0.3ex] 
\hline 
AAPL & 0.2787 &  0.0226 & 0.0114 & 0.0134 \\ 
AMZN & 0.1115  & 0.0174 & 0.0100 & 0.0117 \\
GOOGL & 0.1466 & 0.0139 & 0.0087 & 0.0167 \\
MSFT & 0.2088 & 0.0102 & 0.0072 & 0.0022 \\
XOM & 0.0536 & 0.0036 & 0.0038 & 0.0019 \\
INTC & 0.0143 & 0.0011 & 0.0064 & 0.0042 \\
VZ & 0.0188 & 0.0003 & 0.0025 & 0.0043 \\
KO & 0.0307 & 0.0050 & 0.0027 & --0.0002{ } \\
NFLX & 0.0149 & 0.0236 & 0.0074 & 0.0049 \\
CMCSA & 0.0173 & 0.0053 & 0.0042 & --0.0052{ } \\
NVDA & 0.0446 & 0.0053 & 0.0125 & 0.0116 \\
SBUX & 0.0130 & 0.0053 & 0.0053 & 0.0039 \\
WMT & 0.0474 & 0.0053 & 0.0018 & 0.0002 \\
\hline 
\end{tabular}
\label{BL-largecap-summary} 
\end{table}

\begin{landscape}
\begin{table}[!htbp] \centering
\addtolength{\tabcolsep}{-4pt}{\tiny{\tiny}}
\caption{This table reports the covariance matrix, $\widehat{\Sigma_{}}$.}
\label{table:Sigma}
\begin{tabular}{@{\extracolsep{5pt}} cccccccccccccc}
	\\[-1.8ex]\hline
	\hline \\[-1.8ex]
	& AAPL & AMZN & GOOGL & MSFT & XOM & INTC & VZ & KO & NFLX & CMCSA & NVDA & SBUX & WMT \\
	\hline \\[-1.8ex]
	AAPL & $0.009$ & $0.004$ & $0.004$ & $0.003$ & $0.002$ & $0.003$ & $0.001$ & $0.001$ & $0.002$ & $0.002$ & $0.006$ & $0.002$ & $0.001$ \\
	AMZN & $0.004$ & $0.010$ & $0.004$ & $0.003$ & $0.001$ & $0.003$ & $0.001$ & $0.001$ & $0.006$ & $0.001$ & $0.005$ & $0.003$ & $0.001$ \\
	GOOGL & $0.004$ & $0.004$ & $0.006$ & $0.003$ & $0.001$ & $0.002$ & $0.001$ & $0.001$ & $0.003$ & $0.002$ & $0.004$ & $0.003$ & $0.001$ \\
	MSFT & $0.003$ & $0.003$ & $0.003$ & $0.004$ & $0.001$ & $0.002$ & $0.001$ & $0.001$ & $0.003$ & $0.002$ & $0.004$ & $0.002$ & $0.001$ \\
	XOM & $0.002$ & $0.001$ & $0.001$ & $0.001$ & $0.005$ & $0.002$ & $0.001$ & $0.001$ & $0.0004$ & $0.002$ & $0.002$ & $0.0005$ & $0.0003$ \\
	INTC & $0.003$ & $0.003$ & $0.002$ & $0.002$ & $0.002$ & $0.005$ & $0.001$ & $0.001$ & $0.002$ & $0.002$ & $0.004$ & $0.002$ & $0.001$ \\
	VZ & $0.001$ & $0.001$ & $0.001$ & $0.001$ & $0.001$ & $0.001$ & $0.003$ & $0.001$ & $$-$0.0004$ & $0.002$ & $0.002$ & $0.001$ & $0.001$ \\
	KO & $0.001$ & $0.001$ & $0.001$ & $0.001$ & $0.001$ & $0.001$ & $0.001$ & $0.002$ & $0.0001$ & $0.001$ & $0.001$ & $0.001$ & $0.001$ \\
	NFLX & $0.002$ & $0.006$ & $0.003$ & $0.003$ & $0.0004$ & $0.002$ & $$-$0.0004$ & $0.0001$ & $0.023$ & $0.001$ & $0.005$ & $0.003$ & $$-$0.00004$ \\
	CMCSA & $0.002$ & $0.001$ & $0.002$ & $0.002$ & $0.002$ & $0.002$ & $0.002$ & $0.001$ & $0.001$ & $0.005$ & $0.003$ & $0.002$ & $0.001$ \\
	NVDA & $0.006$ & $0.005$ & $0.004$ & $0.004$ & $0.002$ & $0.004$ & $0.002$ & $0.001$ & $0.005$ & $0.003$ & $0.018$ & $0.003$ & $0.001$ \\
	SBUX & $0.002$ & $0.003$ & $0.003$ & $0.002$ & $0.0005$ & $0.002$ & $0.001$ & $0.001$ & $0.003$ & $0.002$ & $0.003$ & $0.006$ & $0.0005$ \\
	WMT & $0.001$ & $0.001$ & $0.001$ & $0.001$ & $0.0003$ & $0.001$ & $0.001$ & $0.001$ & $$-$0.00004$ & $0.001$ & $0.001$ & $0.0005$ & $0.002$ \\
	\hline \\[-1.8ex]
\end{tabular}
\label{BL-largecap-covar}
\end{table}
\end{landscape}

Suppose the investor has the following views on the expected monthly returns:
\begin{enumerate}
\item Apple will almost surely perform 1.5\% in the following month.
\item Google will outperform Microsoft by 2.5\% in the following month.
\item Verizon will outperform Comcast by 2\% in the following month.
\item Coca-Cola will perform $-$0.5\% in the following month.
\end{enumerate}

Then we formulate an investor's views as follows:
\begin{align}
\boldsymbol{V} | \boldsymbol{M} = {\bf m} \:\sim\: N_{4}(P{\bf m}, \Omega_v)
\end{align}

The above views and the pick matrix can be specified as follows:
\setcounter{MaxMatrixCols}{50}
\begin{align}
{\bf v}_{4 \times 1} \:=\: \begin{pmatrix}
1.5\% \\
2.5\% \\
2\% \\
0.5\%
\end{pmatrix}
\end{align}
\begin{align}
\boldsymbol{\textnormal{P}}_{4 \times 10} \:=\: \begin{pmatrix}
1 & 0 & 0 & 0 & 0 & 0 & 0 & 0 & 0 & 0 \\
0 & 0 & 1 & -1 & 0 & 0 & 0 & 0 & 0 & 0 \\
0 & 0 & 0 & 0 & 0 & 0 & 1 & 0 & 0 & -1 \\
0 & 0 & 0 & 0 & 0 & 0 & 0 & -1 & 0 & 0
\end{pmatrix}
\end{align}

And the following matrix gives a measure of the investors' confidence level on the views. We assume that the investor has a high confidence in the views \textcolor{blue}{(Idzorek, 2004)}.
\begin{align}
\Omega_{v_{4 \times 4}} \:=\: \begin{pmatrix}
1\%^{2} & 0 & 0 & 0 \\
0 & 1\%^{2} & 0 & 0 \\
0 & 0 & 1\%^{2} & 0 \\
0 & 0 & 0 & 1\%^{2}
\end{pmatrix}
\end{align}

The zeros on the off-diagonal elements reflect that the views are uncorrelated.

Suppose an investor solves her portfolio problem by setting the monthly target return to 1.25\%. Then the Black-Litterman optimization problem is given by:
\begin{align}
\underset{w}{\textnormal{min}} \:\:\:\:\:\:\: \boldsymbol{w}^{t}(\Sigma \:+ \:&\Sigma_{BL})\boldsymbol{w} \label{eq:33} \\
\textnormal {subject to} \:\:\:\:\:\:\:\:\: \boldsymbol{w}^{t}\boldsymbol{\mu}_{BL} \:&=\: 1.25 \nonumber\\
\mathbf{1}^{t}\boldsymbol{w} \:&=\: 1. \nonumber
\end{align}

Thus we have:
\begin{align}
\boldsymbol{R}_{\boldsymbol{w}^{Opt}_{BL}} \:\sim\: N(1.25\%, 4.95\%)
\end{align}

Tables~\ref{BL-largecap-port} 
reports the Black-Litterman optimal portfolios given various corresponding monthly target returns $M$ obtained by the optimization problem (\ref{eq:33}).

\begin{table}[!htbp] \small
\caption{\small This table reports the BL optimal portfolios, $\boldsymbol{w}^{Opt}_{BL_{i}}$ given the corresponding monthly target returns, $M_{1} = 0.42\%$, $M_{2} = 0.83\%$, $M_{3} = 1.25\%$, and $M_{4} = 1.67\%$, obtained by the optimization problems (\ref{eq:33}).} 
\centering 
\begin{tabular}{c c c c c} 
\hline 
Ticker & $\boldsymbol{w}^{Opt}_{BL_{1}}$ & $\boldsymbol{w}^{Opt}_{BL_{2}}$ & $\boldsymbol{w}^{Opt}_{BL_{3}}$ & $\boldsymbol{w}^{Opt}_{BL_{4}}$ \\ [0.3ex] 
\hline 
AAPL & 0.039 & 0.090 & 0.143 & 0.196 \\ 
AMZN & --0.014{ } & 0.001 & 0.016 & 0.031 \\
GOOGL & 0.094 & 0.233 & 0.374 & 0.516 \\
MSFT & --0.034{ } & --0.135{ } & --0.237{ } &  --0.340{ } \\
XOM & 0.117 & 0.114 & 0.110 & 0.107 \\
INTC & 0.043 & 0.041 & 0.039 & 0.037 \\
VZ & 0.281 & 0.389 & 0.499 &  0.609 \\
KO & 0.190 & 0.119 & 0.046 & --0.026{ } \\
NFLX & 0.048 & 0.046 & 0.044 & 0.041 \\
CMCSA & --0.103{ } & --0.221{ } & --0.342{ } & --0.463{  } \\
NVDA & --0.030{ } & --0.022{ } & --0.014{ } & --0.006{  } \\
SBUX & 0.069 & 0.065 & 0.061 & 0.057 \\
 WMT & 0.299 & 0.280 & 0.260 & 0.240 \\ 
\hline 
\end{tabular}
\label{BL-largecap-port} 
\end{table}

\subsection{Black-Litterman Allocation in Skew-Normal Markets}
We first perform a statistical procedure, the likelihood ratio test, to see if the skew parameter is statistically significant for the current data. The values of the test are compared with the values of a chi-squared distribution with 13 degrees of freedom which is the dimension of the skew parameter $\boldsymbol{\lambda}_1$ in this example. The results in Table~\ref{likelihoodRatioTest} prove that the skew-normal assumption is more appropriate and promising distribution for our data because it can account for skewness. Table~\ref{MLE-largecap} 
reports the results of the maximum likelihood estimation (MLE) of parameters $\boldsymbol{\mu}, \boldsymbol{\lambda}_{1}, \text{and } \boldsymbol{\tau_1}$ for the 13 equities.

\begin{table}[!htbp]
\small
\caption{\small Likelihood ratio test ($H_{0}:$ $\lambda_1 = 0$)} 
\centering 
\begin{tabular}{c c} 
\\[-1.8ex]\hline
\hline \\[-1.8ex]
& Observed--data \\ [0.3ex] 
\hline 
Log-likelihood normal & 3652.19 \\ 
Log-likelihood skew-normal & 3668.80 \\
2 $\times$ ln(LR) & 33.212 \\
P-value & 0.0015 \\ 
\hline \\[-1.8ex] 
\end{tabular}
\label{likelihoodRatioTest} 
\end{table}

\begin{table}[!htbp]
\small
\caption{\small Maximum likelihood estimates of $\boldsymbol{\mu}$, $\boldsymbol{\lambda}_{1}$, and $\boldsymbol{\tau}_{1}$.}
\centering 
\begin{tabular}{c c c c} 
\\[-1.8ex]\hline
\hline \\[-1.8ex]
Ticker & $\hat{\boldsymbol{\mu}}$ & $\hat{\boldsymbol{\lambda}}_{1}$ & $\hat{\boldsymbol{\tau}}_{1}$ \\ [0.3ex] 
\hline 
AAPL &  0.113 & --1.918{ } & --1.783{ } \\ 
AMZN & 0.077 & --0.648{ } & --0.600{ } \\
GOOGL & 0.074 & --0.938{ } & --0.864{ } \\
MSFT & 0.063  & --0.889{ } & --0.817{ } \\
XOM & 0.043 & --0.715{ } & --0.658{ } \\
INTC & 0.061 & --1.155{ } & --1.062{ } \\
VZ & 0.023 & --0.341{ } & --0.309{ } \\
KO & 0.031 & --0.806{ } & --0.745{ } \\
NFLX & 0.088 & --0.769{ } & --0.706{ } \\
CMCSA & 0.055 & --1.191{ } & --1.099{ } \\
NVDA & 0.138 & --1.744{ } & --1.602{ } \\
SBUX & 0.036 & 0.018 & 0.015 \\
WMT & 0.013 & 0.111 & 0.101 \\ 
\hline \\[-1.8ex] 
\end{tabular}
\label{MLE-largecap} 
\end{table}
\newpage

Assume now that the investor has the same views as earlier. The skew-normal market models are given by:
\begin{align}
\boldsymbol{R}|\boldsymbol{M} = \boldsymbol{m} \:&\sim\: SNT_k(\lambda_0=0, \boldsymbol{\lambda}_1, {\bf m - s}, \Sigma) \\
\boldsymbol{M} \:&\sim\: N_{k}(\pi, \tau \Sigma)
\end{align}

In addition, by Theorem \ref{mainResult}, we have
\begin{align}
\boldsymbol{R}|\boldsymbol{V} = \boldsymbol{v} \:&\sim\: SNT_k(\tau_0=0, \boldsymbol{ \tau}_1, \boldsymbol{\mu}_{BL} - \textbf{s}, \Sigma + \Sigma_{BL})
\end{align}

Now we solve the Black-Litterman allocation problem under the skew-normal assumption as shown in (\ref{eq:23}). For example, we pick one target portfolio monthly return (1.25\%) and six different values of non-spherical variance ($N$ = 0, 0.005, 0.01, 0.015, 0.02, 0.025, 0.03, and 0.04). The results of the skew-normal BL optimal portfolios $\boldsymbol{w}^{Opt}_{skBL,i}$ for $i = 1,...,8$ are reported in Table~\ref{skewBL-largecap-port}. 

\begin{table}[!htbp]
\small
\caption{\small Black-Litterman Skew-Normal Optimal Portfolios Given $M = 1.25\%$ and $N = 0, 0.005, 0.01, 0.015, 0.02, 0.025, 0.03$, and $0.04$} 
\centering 
\begin{tabular}{c c c c c c c c c} 
\\[-1.8ex]\hline
\hline \\[-1.8ex]
Ticker & $\boldsymbol{w}^{Opt}_{N = 0}$ & $\boldsymbol{w}^{Opt}_{N = 0.005}$ & $\boldsymbol{w}^{Opt}_{N = 0.01}$ & $\boldsymbol{w}^{Opt}_{N = 0.015}$ & $\boldsymbol{w}^{Opt}_{N = 0.02}$ &
$\boldsymbol{w}^{Opt}_{N = 0.025}$ & $\boldsymbol{w}^{Opt}_{N = 0.03}$ & $\boldsymbol{w}^{Opt}_{N = 0.04}$ \\ [0.3ex] 
\hline 
AAPL & --0.037{ } & --0.018{ } & 0.002 & 0.022 & 0.042 & 0.061 & 0.081 & 0.120 \\ 
AMZN & --0.036{ } & --0.026{ } & --0.015{ } & --0.004{ } & 0.007 & 0.018 & 0.029 & 0.050 \\
GOOGL & --0.052{ } & --0.007{ } & 0.037 & 0.082 & 0.126 & 0.170 & 0.215 & 0.303\\
MSFT & 0.164 & 0.037 & --0.089{ } & --0.216{ } & --0.342{ } & --0.468{ } & --0.594{ } & --0.847{ } \\
XOM & 0.167 & 0.128 & 0.089 & 0.050 & 0.011 & --0.029{ } & --0.068{ } & --0.146{ } \\
INTC & 0.060 & 0.037 & 0.014 & --0.009{ } & --0.032{ } & --0.055{ } & --0.078{ } & --0.124{ } \\
VZ & 0.186 & 0.224 & 0.261 & 0.299 & 0.336 & 0.374 & 0.411 & 0.486 \\
KO & 0.144 & 0.231 & 0.318 & 0.404 & 0.491 & 0.578 & 0.665 & 0.839 \\
NFLX & 0.027 & 0.048 & 0.069 & 0.090 & 0.111 & 0.132 & 0.153 & 0.196 \\
CMCSA & 0.007 & --0.039{ } & --0.086{ } & --0.133{ } & --0.179{ } & --0.226{ } & --0.273{ } & --0.366{ } \\
NVDA & --0.056{ } & --0.048{ } & --0.039{ } & --0.031{ } & --0.022{ } & --0.014{ } & --0.005{ } & 0.012 \\
SBUX & 0.116 & 0.091 & 0.065 & 0.040 & 0.014 & --0.011{ } & --0.037{ } & --0.086{ } \\
WMT & 0.310 & 0.342 & 0.374 & 0.406 & 0.437 & 0.469 & 0.501 & 0.564 \\
\hline \\[-1.8ex] 
\end{tabular}
\label{skewBL-largecap-port} 
\end{table}
\clearpage

Based on (\ref{eq:19}), (\ref{eq:23}) and (\ref{eq:24}), we could evaluate the portfolio skewness $\tau_1$ and portfolio volatility given different target returns and non-spherical part of the variance. Table~\ref{skBL-largecap-MN} 
reports the portfolios skewness, expected returns and volatilities across different non-spherical variance $N$ ranging from .00 to .05 and different values of monthly target returns $M$ ranging from 0.21\% to 1.67\%. Most importantly, we find that the skew-normal BL optimal allocation provides a lower volatility than that of the classical Black-Litterman given the same target portfolio return except for $M$ = 0.42\%. For example, with fixing 1.25\% as the monthly target return and $N = 0$, the skew-normal BL optimal allocation provides a volatility of 3.552\%, which is less than 4.950\% of the classical BL optimal allocation. We also find the following interesting empirical observations:

\begin{itemize}
\item The portfolios become more negatively skewed as the expected returns of portfolios increase for any given $N$, which suggests that the investors trade a negative skewness for a higher expected return or vice versa.
\item The negative relation between portfolio volatility and portfolio skewness is robust for any given $N$. That is, the investors trade a lower volatility for a higher skewness or vice versa reflecting that stocks with big drops in price are more volatile.
\end{itemize}

\begin{table}[!htbp]
\small
\caption{\small BL Skew-Normal Optimal Portfolio Skewness, Expected Return and Volatility given $N$ = 0, 0.01, 0.02, 0.03, 0.04 and 0.05.} 
\centering 
\begin{tabular}{c c c c c c c c c} 
\hline 
\multicolumn{8}{c}{\textbf{\textit{N}} = \textbf{0}} \\
\hline
& $\boldsymbol{w}_{skBL_{1}}$ & $\boldsymbol{w}_{skBL_{2}}$ & $\boldsymbol{w}_{skBL_{3}}$ & $\boldsymbol{w}_{skBL_{4}}$ & $\boldsymbol{w}_{skBL_{5}}$ & $\boldsymbol{w}_{skBL_{6}}$ &
$\boldsymbol{w}_{skBL_{7}}$ &
$\boldsymbol{w}_{skBL_{8}}$
\\ [0.3ex] 
\hline 
$M$ & 0.21\% & 0.42\% & 0.63\% & 0.83\% & 1.04\% & 1.25\% & 1.46\% & 1.67\% \\
$\tau^{w}_{1}$ & --0.198{ } & --0.242{ } & --0.295{ } & --0.355{ } & --0.424{ } & --0.498{ } & --0.570{ } & --0.632{ } \\ 
Volatility & 5.155\% & 4.741\% & 4.365\% & 4.033\% & 3.757\% & \textcolor{red}{3.552\%} & 3.431\% & 3.401\% \\
\hline \hline 
\multicolumn{8}{c}{\textbf{\textit{N}} = \textbf{0.01}} \\
& $\boldsymbol{w}_{skBL_{1}}$ & $\boldsymbol{w}_{skBL_{2}}$ & $\boldsymbol{w}_{skBL_{3}}$ & $\boldsymbol{w}_{skBL_{4}}$ & $\boldsymbol{w}_{skBL_{5}}$ & $\boldsymbol{w}_{skBL_{6}}$ &
$\boldsymbol{w}_{skBL_{7}}$ &
$\boldsymbol{w}_{skBL_{8}}$ \\ [0.3ex] 
\hline
$M$ & 0.21\% & 0.42\% & 0.63\% & 0.83\% & 1.04\% & 1.25\% & 1.46\% & 1.67\%\\
$\tau^{w}_{1}$ & --0.169{ } & --0.214{ } & --0.262{ } & --0.310{ } & --0.355{ } & --0.392{ } & --0.416{ } & --0.429{ }\\ 
Volatility & 4.659\% & 3.837\% & 3.608\% & 3.459\% & 3.399\% & 3.434\% & 3.560\% & 3.769\% \\
\hline \hline 
\multicolumn{8}{c}{\textbf{\textit{N}} = \textbf{0.02}} \\
& $\boldsymbol{w}_{skBL_{1}}$ & $\boldsymbol{w}_{skBL_{2}}$ & $\boldsymbol{w}_{skBL_{3}}$ & $\boldsymbol{w}_{skBL_{4}}$ & $\boldsymbol{w}_{skBL_{5}}$ & $\boldsymbol{w}_{skBL_{6}}$ &
$\boldsymbol{w}_{skBL_{7}}$ &
$\boldsymbol{w}_{skBL_{8}}$ \\ [0.3ex] 
\hline
$M$ & 0.21\% & 0.42\% & 0.63\% & 0.83\% & 1.04\% & 1.25\% & 1.46\% & 1.67\% \\
$\tau^{w}_{1}$ & --0.106{ } & --0.140{ } & --0.172{ } & --0.199{ } & --0.221{ } & --0.238{ } & --0.249{ } & --0.257{ } \\ 
Volatility & 3.509\% & 3.420\% & 3.425\% & 3.522\%
 & 3.706\% & 3.963\% & 4.280\% & 4.647\% \\
\hline \hline 
\multicolumn{8}{c}{\textbf{\textit{N}} = \textbf{0.03}} \\
& $\boldsymbol{w}_{skBL_{1}}$ & $\boldsymbol{w}_{skBL_{2}}$ & $\boldsymbol{w}_{skBL_{3}}$ & $\boldsymbol{w}_{skBL_{4}}$ & $\boldsymbol{w}_{skBL_{5}}$ & $\boldsymbol{w}_{skBL_{6}}$ &
$\boldsymbol{w}_{skBL_{7}}$ \\ [0.3ex] 
\hline
$M$ & 0.21\% & 0.42\% & 0.63\% & 0.83\% & 1.04\% & 1.25\% & 1.46\% & 1.67\% \\
$\tau^{w}_{1}$ & --0.026{ } & --0.052{ } & --0.075{ } & --0.095{ } & --0.113{ } & --0.127{ } & --0.139{ } & --0.149{ }\\ 
Volatility & 3.505\% & 3.662\% & 3.895\% & 4.193\% & 4.544\% & 4.935\% & 5.357\% & 5.806\% \\
\hline \hline 
\multicolumn{8}{c}{\textbf{\textit{N}} = \textbf{0.04}} \\
& $\boldsymbol{w}_{skBL_{1}}$ & $\boldsymbol{w}_{skBL_{2}}$ & $\boldsymbol{w}_{skBL_{3}}$ & $\boldsymbol{w}_{skBL_{4}}$ & $\boldsymbol{w}_{skBL_{5}}$ & $\boldsymbol{w}_{skBL_{6}}$ &
$\boldsymbol{w}_{skBL_{7}}$ &
$\boldsymbol{w}_{skBL_{8}}$ \\ [0.3ex] 
\hline
$M$ & 0.21\% & 0.42\% & 0.63\% & 0.83\% & 1.04\% & 1.25\% & 1.46\% & 1.67\% \\
$\tau^{w}_{1}$ & 0.036 & 0.013 & --0.008{ } & --0.026{ } & --0.043{ } & --0.057{ } & --0.070{ } & --0.081{ }\\ 
Volatility & 4.123\% & 4.457\% & 4.833\% & 5.244\% & 5.684\% & 6.143\% & 6.619\% & 7.111\% \\
\hline \hline 
\multicolumn{8}{c}{\textbf{\textit{N}} = \textbf{0.05}} \\
& $\boldsymbol{w}_{skBL_{1}}$ & $\boldsymbol{w}_{skBL_{2}}$ & $\boldsymbol{w}_{skBL_{3}}$ & $\boldsymbol{w}_{skBL_{4}}$ & $\boldsymbol{w}_{skBL_{5}}$ & $\boldsymbol{w}_{skBL_{6}}$ &
$\boldsymbol{w}_{skBL_{7}}$ &
$\boldsymbol{w}_{skBL_{8}}$ \\ [0.3ex] 
\hline
$M$ & 0.21\% & 0.42\% & 0.63\% & 0.83\% & 1.04\% & 1.25\% & 1.46\% & 1.67\% \\
$\tau^{w}_{1}$ & 0.075 & 0.054 & 0.035 & 0.018 & 0.002 & --0.012{ } & --0.024{ } & --0.035{ }\\ 
Volatility & 5.143\% & 5.572\% & 6.023\% & 6.492\% & 6.978\% & 7.474\% & 7.979\% & 8.495\% \\
\hline \hline 
\end{tabular}
\label{skBL-largecap-MN} 
\end{table}

Finally, Table~\ref{largecap-vola} reports BL, and skBL portfolio volatility given corresponding monthly target returns, $M = 0.42\%, 0.83\%, 1.25\%$, and $1.67\%$. Table~\ref{SR} reports Sharpe ratios for BL optimal portfolios and skBL optimal portfolios.
\clearpage


\begin{table}[!htbp]
\small
\caption{\small Optimal Portfolio Volatilities. This table reports BL and skBL optimal portfolios volatilities given $N = 0, 0.01, 0.03$, and $0.04$.} 
\centering 
\begin{tabular}{c c c c c c} 
\\[-1.8ex]\hline
\hline \\[-1.8ex] 
Target Return & $BL^{Opt}$ & $skBL^{Opt}_{N=0}$ & $skBL^{Opt}_{N=0.01}$ & $skBL^{Opt}_{N=0.03}$ & $skBL^{Opt}_{N=0.04}$ \\  [1ex] 
\hline 
$0.21\%$ & 3.399\% & 5.155 \% & 4.132\% & 3.505\% & 4.123\% \\ [0.5ex]
$0.42\%$ & 3.483\% & 4.741\% & 3.837\% & 3.662\% & 4.457\% \\ [0.5ex]
$0.63\%$ & 3.702\% & 4.365\% & 3.608\% & 3.895\% & 4.833\% \\ [0.5ex]
$0.83\%$ & 4.036\% & 4.033\% & 3.459\% & 4.193\% & 5.244\% \\ [0.5ex]
$1.04\%$ & 4.461\% & 3.757\% & 3.399\% & 4.544\% & 5.684\% \\ [0.5ex]
$1.25\%$ & 4.950\% & 3.552\% & 3.434\% & 4.935\% & 6.143\% \\ [0.5ex]
$1.46\%$ & 5.487\% & 3.431\% & 3.560\% & 5.357\% & 6.619\% \\ [0.5ex]
$1.67\%$ & 6.062\% & 3.401\% & 3.769\% & 5.806\% & 7.111\% \\ [0.5ex]
$1.88\%$ & 6.661\% & 3.465\% & 4.047\% & 6.273\% & 7.611\% \\ [0.5ex]
$2.08\%$ & 7.279\% & 3.618\% & 4.382\% & 6.755\% & 8.121\% \\ [0.5ex]
\hline 
\end{tabular}
\label{largecap-vola} 
\end{table}

\begin{table}[!htbp]
\small
\caption{\small Sharpe Ratio. This table reports Sharpe ratios for BL and skBL optimal portfolios.} 
\centering 
\begin{tabular}{c c c c c} 
\\[-1.8ex]\hline
\hline \\[-1.8ex] 
Target Return & $BL^{Opt}$ & $skBL^{Opt}_{N=0}$ & $skBL^{Opt}_{N=0.005}$ & $skBL^{Opt}_{N=0.01}$ \\ [1ex] 
\hline 
$1.04\%$ & --0.025{ } & --0.030{ } & --0.032{ } & --0.033{ } \\ [0.5ex]
$1.25\%$ & 0.019 & 0.027 & 0.028 & 0.028 \\ [0.5ex]
$1.46\%$ & 0.056 & 0.089 & 0.089 & 0.086\\ [0.5ex]
$1.67\%$ & 0.085 & 0.151 & 0.146 & 0.136 \\ [0.5ex]
$1.88\%$ & 0.108 & 0.208 & 0.196 & 0.178 \\ [0.5ex]
$2.08\%$ & 0.128 & 0.257 & 0.236 & 0.212\\ [0.5ex]
\hline 
\end{tabular}
\label{SR} 
\end{table}

We find that there is a negative relation between portfolio expected return and portfolio skewness given the values of $N$. In other words, the investors hold negatively (positively) skewed portfolios with higher (lower) expected return. Our empirical results do suggest that the normality assumption in asset returns leads to portfolios that are more risky than in the case when asymmetry (skewness) is explicitly considered in portfolio construction. Unlike the usual mean-variance analysis, there are multiple efficient portfolios in the mean-variance-skewness analysis. Figure \ref{largeCap Optimal Port_1} and \ref{largeCap Optimal Port_2}
represent the frontiers of the feasible skew-normal BL optimal portfolios in the mean-spherical variance-nonspherical variance space (upper panel) and in the mean-variance-skewness space (lower panel) respectively given $0 \le M \le 2.0\%$ and $-0.5 \le N \le 0.5$ and $-2\% \le M \le 2.0\%$ and $-0.5 \le N \le 0.5$.
\clearpage


\begin{figure}%
	\centering
	\subfloat[Feasible monthly optimal portfolios in the mean-spherical variance-nonspherical variance space]{{\includegraphics[width=10cm]{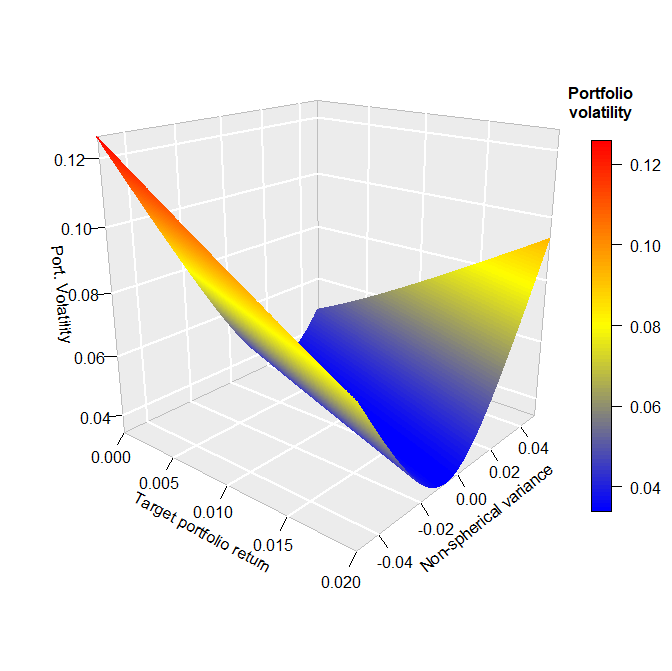}}}%
	\qquad
	\subfloat[Feasible monthly optimal portfolios in the mean-variance-skewness space]{{\includegraphics[width=10cm]{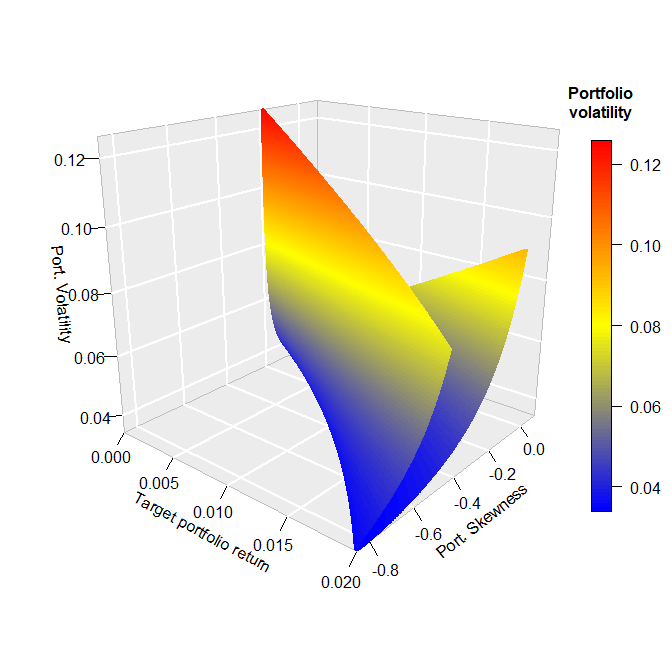}}}%
	\caption{Feasible optimal portfolios frontiers in the mean-spherical variance-nonspherical variance space and in the mean-variance-skewness space ($0 \le M \le 2.0\%$ and $-0.5 \le N \le 0.5$).}%
	\label{largeCap Optimal Port_1}%
\end{figure}
\clearpage

\begin{figure}%
	\centering
	\subfloat[Feasible monthly optimal portfolios in the mean-spherical variance-nonspherical variance space]{{\includegraphics[width=10cm]{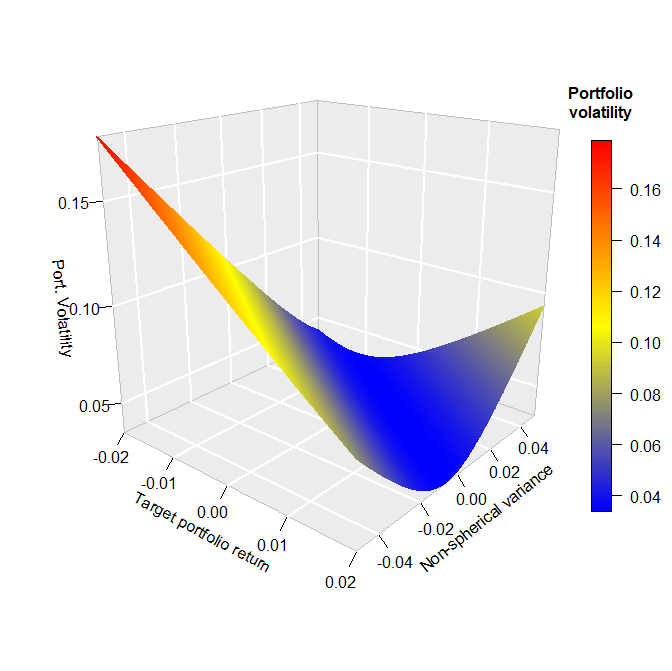}}}%
	\qquad
	\subfloat[Feasible monthly optimal portfoliosin the mean-variance-skewness space]{{\includegraphics[width=10cm]{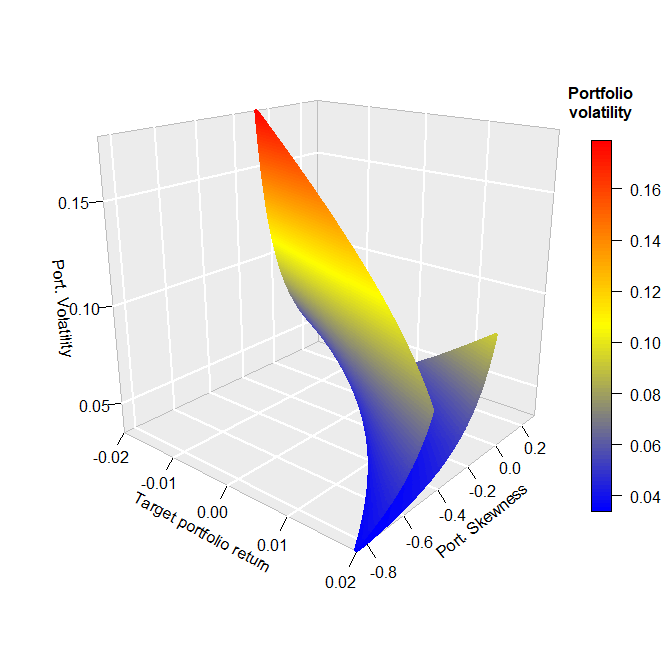}}}%
	\caption{Feasible optimal portfolios frontiers in the mean-spherical variance-nonspherical variance space and in the mean-variance-skewness space ($-2.0\% \le M \le 2.0\%$ and $-0.5 \le N \le 0.5$).}%
	\label{largeCap Optimal Port_2}%
\end{figure}
\clearpage

\section{Conclusion}
The pioneering study conducted by Black and Litterman \textcolor{blue}{(Black \& Litterman, 1990)} provides a framework for how to construct stable mean-variance efficient portfolios. This framework is notably enlightening and aligns with the principles of the Bayesian statistical framework. While the BL model has the capacity to produce portfolios with enhanced stability, it is constrained by the assumption of normality, which represents a limitation. In the context of the Black-Litterman (BL) asset allocation model, skewness is not priced, either because return distribution is assumed to be normal (and thus symmetric) or because investors are assumed to care about only mean and variance. Skew-normal models of portfolio returns have been documented in the recent financial literature due to their enhanced versatility in accommodating distributions with skewness and fat tails. In this paper, we contribute to extend the BL model to the class of skew-normal distributions by departing from the normal market assumption. 

The purpose of this paper is twofold. The first is to extend the classical BL model under the hidden truncation skew-normal distribution as a statistical tool for asset returns. The hidden truncation model provides a flexible family of skewed alternatives to the classical $k$-dimensional normal distribution. It also provides more generalized framework for the Bayesian asset allocation problem. This paper thus presents a theoretical construction of the mean-variance-skewness portfolio optimization based on the multivariate skew-normal distribution. With this model construction, we can solve the optimal asset allocation problems for the extended assumptions of assets returns once the posterior predictive distribution is determined.

The second is to provide the empirical findings derived from our skew-normal BL model. When employing skew-normal returns, we find that the skew-normal BL model provides optimal portfolios with the same expected return but less risk compared to an optimal portfolio of the classical BL model. This finding indicates that when asset returns are assumed to follow a normal distribution, portfolios typically exhibit higher risk, whereas in cases where asymmetry is explicitly accounted for during portfolio construction, risk tends to be lower. It is well-known that assets with large upsides (positive skewness) are overpriced and thus have low expected returns, while assets with large downsides (negative skewness) are underpriced and thus have high expected returns.

The empirical findings of this paper suggest that at any level of $M$ (target portfolio return), a negative trade-off between portfolio volatility and portfolio skewness is robust. This observation suggests that investors may be making a trade-off, opting for lower volatility in exchange for higher skewness, or vice versa. This trade-off indicates that stocks with significant price declines tend to exhibit increased volatility. We also find that the portfolios become more negatively skewed as the expected returns of portfolios increase for given $N$ (non-spherical variance), which suggests that the investors trade a negative skewness for a higher expected return. In other words, the investors hold negatively skewed portfolios with higher expected return or the investors hold positively skewed portfolios with lower expected return. This demonstrates that the investors trade a negative skewness for a higher expected return or vice versa.
\newpage

\section*{References}
\noindent Adcock C. and E. Clark, (1999), ``Beta lives - some statistical perspectives on the capital asset pricing model", The European Journal of Finance, 5, 1999, 213-224.\\
Adcock C., Eling. M. and N. Loperfido, (2015), ``Skewed distributions in finance and actuarial science: a review", The European Journal of Finance, 21, 2015, Issue 21, 1253-1281.\\
Arellano-Valle R. and A. Azzalini, (2008), ``The Centered Parametrization for the Multivariate Skew-Normal Distribution", The Journal of Multivariate Analysis, 99, 2008, 1362-1382.\\
Arnold B. and R Beaver, (2000), ``Hidden Truncation Models", Sankhy$\bar{a}$: The Indian Journal of Statistics, Series A (1961-2002), 62, 2000, Issue 1, 23-35.\\
Arnold B. and R Beaver, (2002), ``Skewed multivariate models related to hidden truncation and / or selective reporting", Test, 11, 2002, Issue 1, 7-54.\\
Arnold B. and R Beaver, (2002), ``Conditionally Specified Multivariate Skewed Distributions", The Indian Journal of Statistics: Series A, 64, 2002, 206-226.\\
Azzalini A., (1985), ``A Class of Distributions Which Includes the Normal Ones",
Scandinavian Journal of Statistics 12, 1985, 171-178.\\
Azzalini A. and A. Dalla Valle, (1996), ``The multivariate skew normal distribution",
Biometrika 83, 1996, 715-726.\\
Azzalini A. and A. Capitanio, (1999), ``Statistical applications of the multivariate skew normal distributions", Journal of the Royal Statistical Society: Series B (Statistical Methodology) 61, 1999, Issue 3, 579-602.\\
Azzalini A., (2005), ``The Skew-normal Distribution and Related Multivariate Families", The Scandinavian Journal of Statistics, Vol 32, 2005, 159-188.\\
Azzalini A. and A. Dalla Valle, (2006), ``On the Unification of Families of Skew-normal Distributions", The Scandinavian Journal of Statistics, Vol 33, 2006, 561-574.\\
Bacmann J-F. and S.M. Benedetti, (2009), ``Optimal Bayesian Portfolios of Hedge Funds", Int. J. Risk Assessment and Management, Vol 11, 2009, 39-58.\\
Barberis N. and M. Huang, (2008), ``Stocks as Lotteries: The Implications of Probability Weighting for Security Prices", American Economic Review, Vol 98, 2008, 2066-2100.\\
Bekart G., Erb C. B., Harvey C. and T. Viskanta, (1998), ``Distributional characteristics of emerging market returns and asset allocation", Journal of Portfolio Management, 1998, Vol 24, 102-116.\\
Best M. and R. Grauer, (1991), ``Sensitivity analysis for mean-variance portfolio problems", Management Science, 1991, Vol 37, 980-989.\\
Black F. and R. Litterman, (1990), ``Asset Allocation: combining investor views with market equilibrium", Goldman Sachs Fixed Income Research, 1990.\\
Black F. and R. Litterman, (1992), ``Global Portfolio Optimization", Financial
Analysts Journal 48:5, September/October 1992, 28-43.\\
Black F., (1993), ``Beta and Return", The Journal of Portfolio Management, 1993, Vol 20, Issue 1, 8-18.\\
Blasi F. S. (2008), ``Bayesian Asset Allocation Using the Skew Normal Distribution", PhD thesis, 2008.\\
Box G. and G. Tiao (1973), Bayesian Inference in Statistical Analysis, Addison-Wesley Publishing Company, 1973.\\
Carmichael B. and A. Co$\ddot{\textnormal{e}}$n, (2013), ``Asset Pricing with Skewed-Normal Return", Finance Research Letters, 2013, Vol 10, 50-57.\\
Conine T. and M. Tamarkin, (1981), ``On Diversification Given Asymmetry in Returns", The Journal of Finance, 1981, Vol 36, 1143-1155.\\
Fabozzi F., S. Focardi and P. Kolm, (2008), ``Incorporating Trading Strategies in the Black-Litterman Framework", Journal of Trading 1(2), 2008, 1-10\\
Fama E., (1963), ``Mandelbrot and the Stable Paretian Hypothesis", The Journal of Business, Vol 36, 1963, 420-429\\
Fang K. and Y. Zhang, (1990), Generalized Multivariate Analysis, Science Press Beijing and Springer-Verlag Heidelberg, 1990.\\
Friend I. and R. Westerfield, (1980), ``Co-skewness and capital asset pricing", the Journal of Finance, Vol 35, 1980, 897-913.\\
Gan Q., (2014), ``Location-scale portfolio selection with factor-recentered skew normal asset returns", Journal of Economic Dynamics and Control 48, 2014, 176-187.\\
Garrett T. and R. Sobel, (1999), ``Gamblers Favor Skewness, not Risk: Further Evidence from United States' Lottery Games", Economic Letters 63, 1999, 85-90.\\
Genton M. G. (ed.), (2004), Skew-Elliptical distributions and their applications: a
journey beyond normality, Chapman $\&$ Hall/CRC, Boca Raton.\\
Giacometti R., M. Bertocch, R. Svetlozar and F. Fzbozzi, (2007), ``Stable distributions in the Black-Litterman approach to asset allocation", Quantitative Finance 7, 2007, Issue 4, 423-433.\\
Golec J. and A. Tamarkin, (1998), ``Bettors love skewness, nor risk, at the horse track", Journal of Political Economy, 1998, Vol 106, 205-225.\\
Harvey C. and A. Siddique, (2000), ``Conditional Skewness in Asset Pricing Tests", The Journal of Finance, 2000, 1263-1296.\\
Harvey C., J. Liechty, M. Liechty and P. M$\ddot{u}$ller, (2010), ``Portfolio selection with higher moments", Quantitative Finance 10, 2010, Issue 5, 469-485.\\
He G. and R. Litterman, (1999), ``The Intuition Behind Black-Litterman Model Portfolios'', Investment Management Research, Goldman Sachs Quantitative Resources Group, 1999, 1-18.\\
Horvath P. and R. Scott, (1980), ``On the Direction of Preference for Moments of Higher Order Than the Variance'', The Journal of Finance, 1999, Vol 35, 915-919.\\
Hu W. and A. Kercheval, (2010), ``Portfolio Optimization for Student t and Skewed t Returns", Quantitative Finance, 2010, Vol 10, 91-105.\\
Idzorek T.M., (2004), ``A Step-by-Step Guide to the Black-Litterman Model. Incorporating User-Specified Confidence Intervals", Zephyr Associates, Inc.\\
Ilmanen A., (2012), ``Do financial markets reward buying or selling insurance and lottery tickets?", Financial Analysts Journal, 2012, Vol 68, 26-36.\\
Ingersoll J., (1987), Theory of Financial Decision Making, Rowman $\&$ Littlefield,
1987.\\
Jondeau E. and M. Rockinger, (2006), ``Optimal Portfolio Allocation under Higher Moments", European Financial Management, 2006, Vol 12, 29-55.\\
Kotz S., (2005), ``Survey of developments in the theory of continuous skewed distributions", International Journal of Statistics, 2005, Vol 63, 225-261.\\
Kraus A. and R. H. Litzenberger, (1976), ``Skewness preference and the valuation of risk assets", Journal of Finance, 1976, 1085-1100.\\
Leland H., (1999), ``Beyond Mean-Variance: Performance measurement in a nonsymmetirical world", Financial Analysts Journal, 1999, 27-36.\\
Liseo B. and N. Loperfido, (2003), ``A Bayesian interpretation of the multivariate skew-normal distribution", Statistics and Probability Letters 61, 2003, 395-401.\\
Mandelbrot B., (1963), ``The variation of certain speculative prices", Journal of Business 36, 1963, 394-419.\\
Markowitz H., (1952), ``Portfolio Selection", The Journal of Finance, Vol 7, 1952, 77-91.\\
Meucci A., (2005), Risk and Asset Allocation, Springer, 2005.\\
Meucci A., (2006a), ``Beyond Black-Litterman: views on non-normal markets",
Risk, February 2006a, 87-92.\\
Meucci A., (2006b), ``Beyond Black-Litterman in Practice: A Five-Step Recipe to Input Views on Non-Normal Markets." Available at SSRN:\\ \href{https://ssrn.com/abstract=872577}{https://ssrn.com/abstract=872577.}\\
Ogundimu E. and J. Hutton, (2015), ``On the extended two-parameter generalized skew-normal distribution", Statistics and Probability Letters 100, 2015, 142-148.\\
Peir$\acute{\textnormal{o}}$ A., (1999), ``Skewness in financial returns", Journal of Banking and Finance 23, 1999, 847-862.\\
Pflug G., (2000), Some Remarks on the Value-at-Risk and the Conditional Value-at-Risk, Probabilistic Constrained Optimization, Kluwer Academic Publishers, 2000,  272-281.\\
Polson, N.G. and B.V. Tew, (2000), ``Bayesian Portfolio Selection: An Empirical Analysis of the S\&P 500 index 1970-1996", J. Business and Econ. Statist., 2000, 164-73.\\
Post T., P. Vliet and H. Levy, (2006), ``Risk Aversion and Skewness Preference", Journal of Banking and Finance 32, 2006, 1178-1187.\\
Pourahmadi M., (2007), ``Construction of Skew-Normal Random Variables",\\ \href{http://citeseerx.ist.psu.edu/viewdoc/summary?doi=10.1.1.211.2167}{http://citeseerx.ist.psu.edu/viewdoc/summary?doi=10.1.1.211.2167.}\\
Qian E. and S. Gorman, (2001), ``Conditional Distribution in Portfolio Theory", Financial Analysts Journal 57, 2001, 44-51.\\
Rama C., R. Deguest and G. Scandolo, (2010), ``Robustness and Sensitivity Analysis of Risk Measurement Procedures", Quantitative Finance 10, 2010, 593-606.\\
Rockafellar R. and S. Uryasev, (2000), ``Optimization of Conditional Value at
Risk", Journal of Risk 2, 2000, 21-41.\\
Roman S., (2004), Introduction to the Mathematics of Finance, Springer, 2004.\\
Ross S., (1978), ``Mutual Fund Separation in Financial Theory - The Separating
Distributions", Journal of Economic Theory 17, 1978.\\
Sahu S.K., D.K. Dey and M. Branco, (2003), ``A New Class of Multivariate Skew Distributions with Application to Bayesian Regression Models", Canad. J. Statist. 31, 2003, Vol 31, 129-150.\\
Samuelson P., (1970), ``The Fundamental Approximation Theorem of Portfolio Analysis in terms of Means, Variances and Higher Moments", The Review of Economic Studies, 1970, Vol 37, 537-542.\\
Sharpe W., (1964), ``Capital asset prices: A theory of market equilibrium under conditions of risk", Journal of Finance, 19, Issue 3, 425-442\\
Shore M., (2005), Hedge Funds: Insights in Performance Measurement, Risk Analysis, and Portfolio Allocation, Wiley $\&$ Sons, 2005.\\
Simaan Y., (1993), ``Portfolio Selection and Asset Pricing Three Parameter Framework", Management Science 5, 1993, 578-587.\\
Singleton J. and J. Wingender, (1986), ``Skewness Persistence in Common Stock Returns", Journal of Financial and Quantitative Analysis, 21, 1986,  Issue 3, 335-341.\\
Sortino F. and R. Van der Meer, (1991), ``Downside Risk", Journal of Portfolio Management, 1991, Vol 17, 27-31.\\
Taleb N., (2004), ``Bleed or Blowup? Why Do We Prefer Asymmetric Payoffs?", Journal of Behavioral Finance, 2004, Vol 5, 2-7.\\
Thaler R.H. and W.T. Ziemba, (1988), ``Anomalies. Parimutuel Betting Markets: Racetracks and Lotteries", Journal of Economic Perspectives, 1988, Vol 2. 161-174.\\
Theil H., (1971), Principles of Econometrics, New York: Wiley and
Sons, 1971.\\
Weil P., (1989), ``The Equity Premium Puzzle and the Risk-Free Rate Puzzle", Journal of Monetary Economics, 1989, Vol 24, 401-421.\\
Xiao Y. and E. Valdez, (2015), ``A Black-Litterman asset allocation model under Elliptical distributions", Quantitative Finance, 2015, Vol 15, 509-519.
\clearpage

\section{Appendix A}
In this Appendix, we report the commonly used results as well as the proofs of the theorems that are presented in the main article. In addition, we also provide a brief development of the hidden truncation skewed-normal distribution.

\subsection{Common Results}
\begin{thm}\label{ShermanMorrisonWoodbury} \textnormal{(Sherman-Morrison-Woodbury Matrix Identity)}
	Suppose A, B, C, D are matrices with the right conditions. Then
	$$
	(A + UBV)^{-1} = A^{-1} - A^{-1}U(B^{-1} + VA^{-1}U)^{-1}VA^{-1}.
	$$
\end{thm}

\begin{thm}
	Assume $\boldsymbol{X} \sim N(\mu, \Sigma)$, where ${\bf \mu}$ is the mean vector and $\Sigma$ is the covariance matrix. Then we have
	$$
	E[\Phi(\boldsymbol{a}^{T} \boldsymbol{X} + b)] = \Phi \left\lbrace \frac{b + \mathbf{a}^{T} {\bf \mu}}{\sqrt{(1 + \mathbf{a}^{T} \Sigma \mathbf{a} )}}\right\rbrace
	$$
	for any constant vector $ \mathbf{a} \in \mathbb{R}^{k}$, and constant $b \in \mathbb{R}$.
\end{thm}

\begin{proof}
	Let Z $\sim N(0,1)$ be a random variable independent of the random vector $\boldsymbol{X}$. Then we have
	\begin{align*}
		E[\Phi(\mathbf{a}^{T} \boldsymbol{X} + b)] &= E[P(Z \le \mathbf{a}^{T} \boldsymbol{X} + b)]\\
		&=  E[P(Z - \mathbf{a}^{T} \boldsymbol{X} \le b)]\\
	\end{align*}
	\noindent Let $Y = Z - \mathbf{a}^{T} \boldsymbol{X}$.
	Then $Y \sim N(-a^{T} \mu, 1 + a^{T} \Sigma a)$. Therefore,
	
	\begin{align*}
		\mathbb{E}[&\Phi(\mathbf{a}^{T} \boldsymbol{X} + b)] = E[P(Y \le b)]\\
		&= E \left[ P \left(  \frac{Y + {\bf a}^{T} {\bf \mu}}{\sqrt{1 + {\bf a}^{T} \Sigma {\bf a}}} \le \frac{b + {\bf a}^{T} {\bf\mu}}{\sqrt{1 + {\bf a}^{T} \Sigma {\bf a}}} \right) \right] \\
		&= E \left[ P \left( Z \le \frac{b + {\bf a}^{T} {\bf \mu}}{\sqrt{1 + {\bf a}^{T} \Sigma {\bf a}}} \right) \right]\\
		&= E \left[ \Phi \left( \frac{b + {\bf a}^{T} {\bf \mu}}{\sqrt{1 + {\bf a}^{T} \Sigma {\bf a}}} \right)  \right]\\
		&= \Phi \left( \frac{b + {\bf a}^{T} {\bf \mu}}{\sqrt{1 + {\bf a}^{T} \Sigma {\bf a}}} \right)
	\end{align*}
\end{proof}
\clearpage

\begin{thm}[Product of Gaussian Densities] \label{ProdGaussian1}
	$$
	\varphi({\bf m};{\bf \pi}, \Omega_1)\varphi({\bf v};P{\bf m}, \Omega_2) =
	\varphi({\bf m};{\bf \mu}_{BL}, \Sigma_{BL})h({\bf v}),
	$$
	where
	\begin{align*}
		\mu_{BL} &= \Sigma_{BL}\left[\Omega_1^{-1}{\bf \pi} + P^T\Omega_2^{-1}{\bf v}\right]\\
		\Sigma_{BL} &= \left[\Omega_1^{-1} + P^T\Omega_2^{-1}P\right]^{-1}
	\end{align*}
	and $h({\bf v})$ is independent of ${\bf m}.$
\end{thm}

\begin{proof}
	Observe that
	\begin{align*}
		\varphi({\bf m};&{\bf \pi}, \Omega_1)\varphi({\bf v};P{\bf m}, \Omega_2)\\
		&\propto \textnormal{exp} \left\lbrace - \frac{1}{2} \:  \underset{(2)}{\underbrace {(\mathbf{m} - \boldsymbol{\pi})^{T} \Omega_1^{-1} (\mathbf{m} - \boldsymbol{\pi}) + (\mathbf{v} - P \mathbf{m})^{T} \Omega_2^{-1} (\mathbf{v} - P \mathbf{m})}} \right\rbrace \\
	\end{align*}
	We can rewrite (2) as
	\begin{align*}
		(\mathbf{m} - &\boldsymbol{\pi})^{T} \Omega_1^{-1} (\mathbf{m} - \boldsymbol{\pi}) + (\mathbf{v} - P \mathbf{m})^{T} \Omega_2^{-1} (\mathbf{v} - P \mathbf{m}) \\
		&= \mathbf{m}^{T}\Omega_1^{-1} \mathbf{m} - 2 \mathbf{m}^{T} \Omega_1^{-1} \boldsymbol{\pi} + \boldsymbol{\pi}^{T} \Omega_1^{-1} \boldsymbol{\pi} + \mathbf{v}^{T} \Omega_2^{-1} \mathbf{v}  - 2 (P \mathbf{m})^{T} \Omega_2^{-1} \mathbf{v} \\
		&\quad \quad + (P \mathbf{m})^{T} \Omega_2^{-1} (P \mathbf{m}) \\
		&= \mathbf{m}^{T} [\Omega_1^{-1} + P^T \Omega_2^{-1} P] \mathbf{m} - 2 \mathbf{m}^{T} [\Omega_1^{-1} \boldsymbol{\pi} + P^T \Omega_2^{-1} \mathbf{v}] \\
		&\quad \quad + \underset{(V_{1})}{\underbrace {\boldsymbol{\pi}^{T} \Omega_1^{-1} \boldsymbol{\pi} + \mathbf{v}^{T} \Omega_2^{-1} \mathbf{v}}} \\
		&= \mathbf{m}^{T} [\Omega_1^{-1} + P^T \Omega_2^{-1} P] \mathbf{m} - 2 \mathbf{m}^{T} A \underset{(b)}{\underbrace {A^{-1} [\Omega_1^{-1} \boldsymbol{\pi} + P^T \Omega_2^{-1} \mathbf{v}]}} + V_{1} \\
		&= \mathbf{m}^{T} A \mathbf{m} - 2 \mathbf{m}^{T} A b + b^{T} A b \: \underset{(V_{2})}{\underbrace {- \: b^{T} A b + V_{1}}} \\
		&= (\mathbf{m} - b)^{T} A (\mathbf{m} - b) + V_{2}
	\end{align*}
	where $A = \Omega_1^{-1} + P \Omega_2^{-1} P.$ This completes the proof since ${\bf b} = {\bf \mu}_{BL}.$\\
\end{proof}
\clearpage

\begin{thm}\label{ProdGaussian2} \textnormal{(Product of Gaussian Densities).}
	$$
	\varphi({\bf x};{\bf m}_1, \Omega_1)\varphi({\bf m}_1;{\bf m}_2, \Omega_2) =
	\varphi({\bf x};{\bf m}_2, \Omega_1+\Omega_2)
	\varphi({\bf m}_1; {\bf z(x,m_2)},\Delta),
	$$
	where $\Delta = (\Omega_1^{-1}+\Omega_2^{-1})^{-1}$ \text{ and } ${\bf z(x,m_2)} = \Delta(\Omega_1^{-1}{\bf x} + \Omega_2^{-1}{\bf m}_2).$
\end{thm}

\begin{proof}
	\begin{align}
		\varphi(\mathbf{x}; &\mathbf{m}_{1}, \Omega_{1}) \cdot \varphi(\mathbf{m}_{1}; \mathbf{m}_{2}, \Omega_{2}) \nonumber \\
		&= \textnormal{exp}\left\lbrace - \frac{1}{2} \: [ \underset{(1)} {\underbrace { (\mathbf{x} - \mathbf{m}_1)^{T} \Omega^{-1}_{1} (\mathbf{x} - \mathbf{m}_1) + (\mathbf{m}_{1} - \mathbf{m}_{2})^{T} \Omega^{-1}_{2} (\mathbf{m}_{1} - \mathbf{m}_{2}) }} ] \right\rbrace \nonumber
	\end{align}
	Then we can write (1) as
	\begin{align*}
		\mathbf{x}^{T} &\Omega^{-1}_{1} \mathbf{x} - 2\mathbf{m}_{1}^{T} \Omega^{-1}_{1} \mathbf{x} + \mathbf{m}_{1}^{T} \Omega^{-1}_{1} \mathbf{m}_{1} +  \mathbf{m}_{1}^{T} \Omega^{-1}_{2} \mathbf{m}_{1} - 2 \mathbf{m}_{1}^{T} \Omega^{-1}_{2} \mathbf{m}_{2} +  \mathbf{m}_{2}^{T} \Omega^{-1}_{2} \mathbf{m}_{2} \\
		&= \mathbf{m}_{1}^{T} [\Omega^{-1}_{1} + \Omega^{-1}_{2}] \mathbf{m}_1 - 2 \mathbf{m}_{1}^{T} [\Omega^{-1}_{1} \mathbf{x} + \Omega^{-1}_{2} \mathbf{m}_{2}] + \mathbf{x}^{T} \Omega^{-1}_{1} \mathbf{x} +  \mathbf{m}_{2}^{T} \Omega^{-1}_{2} \mathbf{m}_{2} \\
		&= \mathbf{m}_{1}^{T} \Delta^{-1} \mathbf{m}_1 - 2 \mathbf{m}_{1}^{T} [\Omega^{-1}_{1} + \Omega^{-1}_{2}] [\Omega^{-1}_{1} + \Omega^{-1}_{2}]^{-1} [\Omega^{-1}_{1} \mathbf{x} + \Omega^{-1}_{2} \mathbf{m}_{2}] + \mathbf{z}^{T} \Delta^{-1} \mathbf{z} \\
		&\:\:\: \underset{(\eta)} {\underbrace {- \mathbf{z}^{T} \Delta^{-1} \mathbf{z} + \mathbf{x}^{T} \Omega^{-1}_{1} \mathbf{x} +  \mathbf{m}_{2}^{T} \Omega^{-1}_{2} \mathbf{m}_{2}}} \\
		&= (\mathbf{m}_{1} - \mathbf{z})^{T} \Delta^{-1} (\mathbf{m}_{1} - \mathbf{z}) + \eta
	\end{align*}
	Since
	$$
	\mathbf{z}^{T} \Delta^{-1} \mathbf{z} =
	({\bf x}^T\Omega_1^{-1} + {\bf m}_2^T\Omega_2^{-1})\Delta
	(\Omega_1^{-1}{\bf x} + \Omega_2^{-1}{\bf m}_2),
	$$
	We can also rewrite $\eta$ as
	\begin{align*}
		\eta &= \mathbf{x}^{T} \Omega^{-1}_{1} \mathbf{x} +  \mathbf{m}_{2}^{T} \Omega^{-1}_{2} \mathbf{m}_{2} - \mathbf{z}^{T} \Delta^{-1} \mathbf{z} \\
		&= \mathbf{x}^{T} \Omega^{-1}_{1} \mathbf{x} +  \mathbf{m}_{2}^{T} \Omega^{-1}_{2} \mathbf{m}_{2} - [\mathbf{x}^{T}  \Omega^{-1}_{1} + \mathbf{m}_{2}^{T} \Omega^{-1}_{2}] \Delta [\Omega^{-1}_{1} \mathbf{x} + \Omega^{-1}_{2} \mathbf{m}_{2}] \\
		&= \mathbf{x}^{T} \Omega^{-1}_{1} \mathbf{x} + \mathbf{m}_{2}^{T} \Omega^{-1}_{2} \mathbf{m}_{2} - [\mathbf{x}^{T}  \Omega^{-1}_{1} \Delta \Omega^{-1}_{1} \mathbf{x} + 2\mathbf{x}^{T}  \Omega^{-1}_{1} \Delta \Omega^{-1}_{2} \mathbf{m}_{2} + \mathbf{m}^{T}_{2} \Omega^{-1}_{2} \Delta \Omega^{-1}_{2} \mathbf{m}_{2}] \\
		&= \mathbf{x}^{T} [\Omega^{-1}_{1} - \Omega^{-1}_{1}  \Delta \Omega^{-1}_{1}] \mathbf{x} + \mathbf{m}^{T}_{2} [\Omega^{-1}_{2} - \Omega^{-1}_{2} \Delta \Omega^{-1}_{2}] \mathbf{m}_{2} - 2 \mathbf{x}^{T} \Omega^{-1}_{1} \Delta \Omega^{-1}_{2} \mathbf{m}_{2}
	\end{align*}
	
	Using the Sherman-Morrison-Woodbury matrix identity, we obtain
	\begin{align*}
		\Omega^{-1}_{1} - \Omega^{-1}_{1} \Delta \Omega^{-1}_{1} \:\:=&\:\:\: \Omega^{-1}_{1}  - \Omega^{-1}_{1} (\Omega^{-1}_{1}  + \Omega^{-1}_{2})^{-1} \Omega^{-1}_{1} \:\:=\:\: (\Omega_{1} + \Omega_{2})^{-1} \\
		\Omega^{-1}_{2} - \Omega^{-1}_{2} \Delta \Omega^{-1}_{2} \:\:=&\:\:\: \Omega^{-1}_{2}  - \Omega^{-1}_{2} (\Omega^{-1}_{1}  + \Omega^{-1}_{2})^{-1} \Omega^{-1}_{2} \:\:=\:\: (\Omega_{1} + \Omega_{2})^{-1}
	\end{align*}
	We then have
	\begin{align}
		\eta \:\:=&\:\:\: \mathbf{x}^{T} (\Omega_{1} + \Omega_{2})^{-1} \mathbf{x} + \mathbf{m}^{T}_{2} (\Omega_{1} + \Omega_{2})^{-1} \mathbf{m}_{2} - 2 \mathbf{x}^{T} (\Omega_{1} + \Omega_{2})^{-1} \mathbf{m}_{2} \nonumber \\
		\:\:=&\:\:\: (\mathbf{x} - \mathbf{m}_{2})^{T} (\Omega_{1} + \Omega_{2})^{-1} (\mathbf{x} - \mathbf{m}_{2})
	\end{align}
	This completes the proof.
\end{proof}
\clearpage

\subsection{Multivariate Hidden Truncation Skew-Normal Distribution}
To extend the univariate hidden truncation skew-normal, \textcolor{blue}{Arnold and Beaver (2000)} started with $W_1, W_2, \ldots, W_k$ and $U$ i.i.d. N(0,1). Let ${\bf W} = (W_1, W_2, \ldots, W_k)^T$. Then for any $\lambda_0 \in \mathbb{R}$ and ${\bf \lambda_1} \in \mathbb{R}^k$, consider the event $A = [\lambda_0 + {\bf \lambda_1}^T {\bf W} > U].$ Then the joint conditional distribution of ${\bf W}$ and $U$ given $A$ is
\begin{align*}
	P({\bf W} \le &{\bf w}, U \le u | A) = \frac{1}{P(A)}P({\bf W} \le {\bf w}, U \le u, A)\\
	&= \frac{1}{P(A)}\int_{R^{k+1}}P({\bf W} \le {\bf w}, U \le u | {\bf W} = {\bf x}, U = s) d{\bf x} ds\\
	&= \frac{1}{P(A)}\int_{-\infty}^{\bf w}\int_{-\infty}^u I_{[\lambda_0 + {\bf \lambda}_1^T{\bf x} > s]}f_{\bf W}({\bf x})f_U(s)d{\bf x} ds
\end{align*}
Thus, the joint conditional density of ${\bf W}$ and $U$ given $A$ is
\begin{align}
	f_{\boldsymbol{W},U|A}(\boldsymbol{w},u) = \frac{1}{P(A)} \prod_{i=1}^{k} \varphi(w_{i}) \varphi(u) I(\lambda_{0} + \boldsymbol{\lambda}_{1}^{t} \boldsymbol{w} > u)
\end{align}
By integrating with respect to u we get
\begin{align*}
	f_{\boldsymbol{W}|A}(\boldsymbol{w}) &= \frac{1}{P(A)}\int_R \prod_{i=1}^{k} \varphi(w_{i})\varphi(u)I_{[\lambda_0 + {\bf \lambda}^T_1{\bf w}> u]} du\\
	&= \frac{1}{P(A)}\prod_{i=1}^{k} \varphi(w_{i}) E [I_{[\lambda_{0} + \boldsymbol{\lambda}_{1}^{t} \boldsymbol{w} > U]}]\\
	&= \frac{1}{P(A)}\prod_{i=1}^{k} \varphi(w_{i}) P[\lambda_{0} +\boldsymbol{\lambda}_{1}^{t} \boldsymbol{w} > U]\\
	&= \frac{1}{P(A)}\prod_{i=1}^{k} \varphi(w_{i}) \Phi(\lambda_{0} +\boldsymbol{\lambda}_{1}^{t} \boldsymbol{w})
\end{align*}
In addition, since
$$
P(A) = P(U - {\bf \lambda_1}^T {\bf W} < \lambda_0)
=  \Phi \left( \frac{\lambda_{0}}{\sqrt{1 + \boldsymbol{\lambda}_{1}^{t} \boldsymbol{\lambda}_{1}}} \right),
$$
we have
$$
f_{\boldsymbol{W}|A}(\boldsymbol{w}) = \prod_{i=1}^{k} \varphi(w_{i}) \Phi(\lambda_{0} + \boldsymbol{\lambda}_{1}^{t} \boldsymbol{w}) \Big/ \Phi \left( \frac{\lambda_{0}}{\sqrt{1 + \boldsymbol{\lambda}_{1}^{t} \boldsymbol{\lambda}_{1}}} \right).
$$
Arnold and Beaver \textcolor{blue}{(Arnold $\&$ Beaver, 2000)} take this as the definition of the multivariate hidden truncation skew-normal.
\begin{defn}
	A random vector ${\bf X}$ is said to have a hidden truncation skew-normal with parameters $\lambda_0, {\bf \lambda}_1$ if its density is
	\begin{align}
		f_{\bf X}({\bf x}) = \varphi_k({\bf x};{\bf 0}, I)\Phi(\lambda_{0} + \boldsymbol{\lambda}_{1}^{T} {\bf x}) \Big/ \Phi \left( \frac{\lambda_{0}}{\sqrt{1 + \boldsymbol{\lambda}_{1}^{T} \boldsymbol{\lambda}_{1}}} \right).
	\end{align}
	In this case, we write ${\bf X} \sim SNT_k(\lambda_0, {\bf \lambda}_1).$
\end{defn}

\begin{thm}
	Suppose ${\bf X} \sim SNT_k(\lambda_0, {\bf \lambda_1})$, then the moment generating function is
	\begin{align}
		M(\boldsymbol{s}) &= e^{\boldsymbol{s}^{T} \boldsymbol{s} / 2}\: \Phi \left( \frac{\lambda_{0} + \boldsymbol{\lambda}_{1}^{T} \boldsymbol{s}}{\sqrt{1 + \boldsymbol{\lambda}_{1}^{T} \boldsymbol{\lambda}_{1}}}\right)\:\Bigg/ \:\Phi\left( \frac{\lambda_{0}}{\sqrt{1 + \boldsymbol{\lambda}_{1}^{T} \boldsymbol{\lambda}_{1}}}\right).
	\end{align}
\end{thm}

\begin{proof}
	For ${\bf s} \in \mathbb{R}^k,$ we have
	\begin{align*}
		M_{\bf X}({\bf s}) &= \frac{1}{\Phi \left( \frac{\lambda_{0}}{\sqrt{1 + \boldsymbol{\lambda}_{1}^{T} \boldsymbol{\lambda}_{1}}} \right)}
		\int_{\mathbb{R}^k} e^{{\bf s}^T{\bf x}} \varphi_k({\bf x}; {\bf 0}, I)\Phi(\lambda_0 + {\bf \lambda}_1^T{\bf x})d{\bf x}\\
		&= \frac{e^{{\bf s}^T{\bf s}}}{\Phi \left( \frac{\lambda_{0}}{\sqrt{1 + \boldsymbol{\lambda}_{1}^{T} \boldsymbol{\lambda}_{1}}} \right)}
		\int_{\mathbb{R}^k} \varphi_k({\bf x}; {\bf s}, I)\Phi(\lambda_0 + {\bf \lambda}_1^T{\bf x})d{\bf x}\\
		&= \frac{e^{{\bf s}^T{\bf s}}}{\Phi \left( \frac{\lambda_{0}}{\sqrt{1 + \boldsymbol{\lambda}_{1}^{T} \boldsymbol{\lambda}_{1}}} \right)}
		\mathbb{E}[\Phi(\lambda_0 + {\bf \lambda}_1^T{\bf U})
	\end{align*}
	where ${\bf U} \sim N_k({\bf s}, I).$ Applying Theorem 4.8 in Appendix B, we obtain the result.
\end{proof}

The mean and the covariance matrix of a random vector could be obtained by differentiating the cumulant function $K(t) = \text{ln}(M(t)).$ For ${\bf X} \sim SNT_k(\lambda_0, {\boldsymbol \lambda}_1),$ its cumulant function is
$$
K(t) = \sum_{i =1}^{k} t^{2}_{i} + \textnormal{ln} \:  \Phi \left\lbrace \frac{\lambda_{0} + \boldsymbol{\lambda}_{1}^{t} \textbf{t}}{\sqrt{1 + \boldsymbol{\lambda}_{1}^{t} \boldsymbol{\lambda}_{1}}}\right\rbrace - c.
$$
Differentiating with respect to $t_{i}$ and setting $\boldsymbol{t} = \mathbf{0}$, we have
\begin{align}
	E(X_{i}) = \frac{\lambda_{1i}}{\sqrt{1 + \boldsymbol{\lambda}_{1}^{t} \boldsymbol{\lambda}_{1}}} \frac{ \varphi \left( \frac{\lambda_{0}}{\sqrt{1 + \boldsymbol{\lambda}_{1}^{t} \boldsymbol{\lambda}_{1}}} \right) }{\Phi \left( \frac{\lambda_{0}}{\sqrt{1 + \boldsymbol{\lambda}_{1}^{t} \boldsymbol{\lambda}_{1}}}\right) }.
\end{align}
and with further differentiation we have the variances and covariances.
\begin{align*}
	\textnormal{Cov} (X_{i},X_{j}) = \delta_{ij} - \frac{\lambda_{1i} \lambda_{1j}}{(1 + \boldsymbol{\lambda}_{1}^{t} \boldsymbol{\lambda}_{1})} \left[ \frac{\lambda_{0}}{\sqrt{1 + \boldsymbol{\lambda}_{1}^{t} \boldsymbol{\lambda}_{1}}} h (\lambda_{0}, \boldsymbol{\lambda}_{1}) + h^{2} (\lambda_{0}, \boldsymbol{\lambda}_{1})\right],
\end{align*}
where $\delta_{ij}$ is the Kronecker delta symbol, and where
\begin{align}
	h (\lambda_{0}, \boldsymbol{\lambda}_{1}) = \frac{\varphi \left( \frac{\lambda_{0}}{\sqrt{1 + \boldsymbol{\lambda}_{1}^{t} \boldsymbol{\lambda}_{1}}}\right) }{\Phi \left(\frac{\lambda_{0}}{\sqrt{1 + \boldsymbol{\lambda}_{1}^{t} \boldsymbol{\lambda}_{1}}}\right) }.
\end{align}

The general multivariate hidden truncation skew normal can be obtained by introducing the location and scale parameters. In particular, suppose $\boldsymbol{\mu} \in \mathbb{R}^k$ and ${\Sigma}$ is a $k\times k$ positive definite matrix. With ${\bf Y} \sim SNT_k(\lambda_0, \boldsymbol{\lambda}_1),$ consider ${\bf X} = \boldsymbol{\mu} + \Sigma^{1/2}{\bf Y}.$ With ${J}$ denote the Jacobian of the transformation, we have
\begin{align}
	f_{\boldsymbol{X}}(\boldsymbol{x}) \:&=\: f_{\boldsymbol{Y}}(\Sigma^{-1/2}(\mathbf{x} - \boldsymbol{\mu})) \cdot \abs{J}^{-1} \nonumber \\
	&= \frac{1}{\Phi \left(\frac{\lambda_{0}}{\sqrt{1 + \boldsymbol{\lambda}_{1}^{t} \boldsymbol{\lambda}_{1}}}\right) }\varphi(\Sigma^{-1/2}(\boldsymbol{x} - \boldsymbol{\mu}); 0, I) \cdot \Phi(\lambda_{0} + \boldsymbol{\lambda}^{t}_{1}\Sigma^{-1/2}(\boldsymbol{x} - \boldsymbol{\mu})) \cdot \abs{J}^{-1} \nonumber \\
	&=\: \frac{1}{\Phi \left(\frac{\lambda_{0}}{\sqrt{1 + \boldsymbol{\lambda}_{1}^{t} \boldsymbol{\lambda}_{1}}}\right) } \cdot \frac{1}{(2 \pi)^{k/2} |\Sigma|^{1/2}} \: \textnormal{exp}\left\lbrace - \frac{1}{2} (\boldsymbol{x} - \boldsymbol{\mu})^{t} \Sigma^{- 1/2} \Sigma^{- 1/2}(\boldsymbol{x} - \boldsymbol{\mu}) \right\rbrace \cdot \nonumber \\
	&\:\:\:\:\:\:\: \times \Phi(\lambda_{0} + \boldsymbol{\lambda^{t}_{1}} \Sigma^{- 1/2} (\boldsymbol{x} - \boldsymbol{\mu})) \nonumber \\
	\:& =  \frac{1}{{\Phi \left(\frac{\lambda_{0}}{\sqrt{1 + \boldsymbol{\lambda}_{1}^{t} \boldsymbol{\lambda}_{1}}}\right) }} \cdot \varphi(\boldsymbol{x}; \boldsymbol{\mu}, \Sigma) \cdot \Phi(\lambda_{0} + \boldsymbol{\lambda^{t}_{1}} \Sigma^{- 1/2} (\boldsymbol{x} - \boldsymbol{\mu})) \nonumber \label{eq:7} \\
\end{align}
In this case, we write ${\bf X}\sim SNT_k(\lambda_0, \boldsymbol{\lambda}_1, \boldsymbol{\mu}, \Sigma).$

\begin{thm}
	\label{mgfSNT}
	\textnormal{Suppose} $\boldsymbol{X} \sim SNT_{k}(\lambda_{0}, \boldsymbol{\lambda}_{1}, \boldsymbol{\mu}, \Sigma)$. Then
	\begin{align}
		M_{\bf X}({\bf s}) \:=\: \textnormal{exp} \left\lbrace {\bf s}^{T} \mu + \frac{1}{2} {\bf s}^{T} \Sigma {\bf s} \right\rbrace \cdot \frac{\Phi\left( \frac{\lambda_{0} + \boldsymbol{\lambda}^{T}_{1} \Sigma^{\frac{1}{2}} {\bf s}}{\sqrt{1 + \boldsymbol{\lambda}^{T}_{1} \boldsymbol{\lambda}_{1}}}\right) }{\Phi \left( \frac{\lambda_{0}}{\sqrt{1 + \boldsymbol{\lambda}^{T}_{1} \boldsymbol{\lambda}_{1}}}\right) } \label{eq:8}
	\end{align}
\end{thm}

\begin{proof}
	The proof is straightforward by recognizing that ${\bf X} = {\bf \mu} + \Sigma^{1/2}{\bf Y}$, where ${\bf Y} \sim SNT_k(\lambda_0, {\bf \lambda_1}).$ Since, for ${\bf s} \in \mathbb{R}^k$,
	$$
	M_{\bf X}({\bf s}) = e^{{\bf s}^T{\bf \mu}}M_{\bf Y}(\Sigma^{\frac{1}{2}}{\bf s})
	$$
	and the moment generating function for ${\bf Y}$ was obtained earlier, the result follows.
\end{proof}

The mean and the covariance of ${\bf X}$ can be obtained similarly. In particular, let ${\bf X} \sim SNT_k(\lambda_0, {\bf \lambda_1},  \boldsymbol{\mu}, \Sigma)$, then
\begin{align}
	\mathbb{E}({\bf X}) &= \boldsymbol{\mu} + \Sigma^{1/2}\mathbb{E}({\bf Y}) \label{eq:9}\\
	&= \boldsymbol{\mu} + h(\lambda_0, {\bf \lambda}_1)
	\frac{\Sigma^{1/2}{\bf \lambda}_1}{\sqrt{1 + \lambda_1^t\lambda_1}} \label{eq:10}
\end{align}
where $h (\lambda_{0}, \boldsymbol{\lambda}_{1}) = \frac{\varphi \left( \frac{\lambda_{0}}{\sqrt{1 + \boldsymbol{\lambda}_{1}^{t} \boldsymbol{\lambda}_{1}}}\right) }{\Phi \left(\frac{\lambda_{0}}{\sqrt{1 + \boldsymbol{\lambda}_{1}^{t} \boldsymbol{\lambda}_{1}}}\right) }$, and
$\text{Cov}({\bf X}) = \Sigma^{1/2}\text{Cov}({\bf Y})\Sigma^{1/2}.$
\\~\\

\noindent \textbf{Proof of Theorem \ref{Linear Transformation}}
\begin{proof}
	Let $s \in \mathbb{R}^{m}$ and  $A_{m \times n}$.
	\begin{align}
		M_{\bf Y}(s) \:&=\: \mathbb{E} \: [ e^{s^{T} {\bf Y}} ] \:=\: \mathbb{E} \: [ e^{s^{T} ({\bf b} + AX)} ] \nonumber \\
		&=\: e^{s^T{\bf b}}M_{X}(A^{T}s) \nonumber \\
		&=\: e^{s^T{\bf b}}e^{(A^{T}s)^{T} \mu + \frac{1}{2} (A^{T}s)^{T} \Sigma (A^{T}s)} \cdot \frac{\Phi \left( \frac{\lambda_{0} + \lambda^{T}_{1} \Sigma^{\frac{1}{2}} (A^{T} s)}{\sqrt{1 + \lambda^{T}_{1} \lambda_{1}}}\right) }{\Phi \left( \frac{\lambda_{0}}{\sqrt{1 + \lambda^{T}_{1} \lambda_{1}}}\right) } \nonumber \\
		&=\: e^{s^{T} ({\bf b} + A \mu) + \frac{1}{2} s^{T} (A \Sigma A^{T}) s} \cdot  \frac{ \Phi \left( \frac{\lambda_{0}}{\sqrt{1 + \lambda^{T}_{1} \lambda_{1}}} + \frac{s^{T} A \Sigma^{\frac{1}{2}} \lambda_{1}}{\sqrt{1 + \lambda^{T}_{1} \lambda_{1}}} \right)}{\Phi \left( \frac{\lambda_{0}}{\sqrt{1 + \lambda^{T}_{1} \lambda_{1}}} \right) }
	\end{align}
	To show ${\bf Y} \sim SNT_{m}({\bf \mu}_{y}, \Sigma_{y}, \tau_{0}, {\bf \tau}_{1})$, we need to show:
	\begin{align*}
		M_{\bf Y}(s) \:=\: e^{s^{T} ({\bf \mu}_{y}) + \frac{1}{2} s^{T} \Sigma_{y} s} \cdot  \frac{ \Phi \left( \frac{\tau_{0} + s^{T} \Sigma_{y}^{\frac{1}{2}} {\bf \tau}_{1}}{\sqrt{1 + {\bf \tau}^{T}_{1} {\bf \tau}_{1}}} \right)}{\Phi \left( \frac{\tau_{0}}{\sqrt{1 + {\bf \tau}^{T}_{1} {\bf \tau}_{1}}} \right) }
	\end{align*}
	First, we can write:
	\begin{align}
		s^{T} A \Sigma^{\frac{1}{2}} \lambda_{1} \:&=\: s^{T} \Sigma_{y}^{\frac{1}{2}} \underset {\lambda_{*}} {\underbrace{\Sigma_{y}^{-\frac{1}{2}} A \Sigma^{\frac{1}{2}} \lambda_{1}}} \nonumber \\
		&=\: s^{T} \Sigma_{y}^{\frac{1}{2}} \lambda_{*}
	\end{align}
	Then
	\begin{align}
		\lambda^{T}_{*} \lambda_{*} \:&=\: \lambda^{T}_{1} \Sigma^{\frac{1}{2}} A^{T}\Sigma_{y}^{-\frac{1}{2}} \cdot \Sigma_{y}^{-\frac{1}{2}} A \Sigma^{\frac{1}{2}} \lambda_{1}  \nonumber \\
		&=\: \lambda^{T}_{1} \Sigma^{\frac{1}{2}} A^{T}\Sigma_{y}^{-1} A \Sigma^{\frac{1}{2}} \lambda_{1} \nonumber \\
		&=\: \lambda^{T}_{1} H \Sigma_{y}^{-1} H^{T} \lambda_{1},
	\end{align}
	where $H = \Sigma^{\frac{1}{2}} A^{T}$. In addition, since
	\begin{align}
		1 + \lambda^{T}_{1} \lambda_{1} \:&=\: 1 + \lambda^{T}_{1} \lambda_{1} + \lambda^{T}_{*} \lambda_{*} - \lambda^{T}_{*} \lambda_{*} \nonumber \\
		&=\: 1 + \lambda^{T}_{*} \lambda_{*} + \lambda^{T}_{1} \left[ I - H \Sigma_{y}^{-1} H^{T} \right] \lambda_{1},
	\end{align}
	thus
	$$
	\frac{1 + \lambda_1^T\lambda_1}{1 + \lambda_1^T\left[ I - H \Sigma_{y}^{-1} H^{T} \right] \lambda_1}
	= 1 + \frac{\lambda_*^T\lambda_*}{1 + \lambda_1^T\left[ I - H \Sigma_{y}^{-1} H^{T} \right] \lambda_1}
	= 1 + \tau_1^T \tau_1
	$$
	where $\tau_1$ as defined in the theorem. As a result, we have
	\begin{align*}
		\frac{\lambda_{0}}{\sqrt{1 + \lambda^{T}_{1} \lambda_{1}}} \:&=\: \frac{ \frac{\lambda_{0}}{\sqrt{1 + \lambda^{T}_{1} \left[ I - H \Sigma_{y}^{-1} H^{T} \right] \lambda_{1}}}}{\frac{\sqrt{1 + \lambda^{T}_{1} \lambda_{1}}}{\sqrt{1 + \lambda^{T}_{1} \left[ I - H \Sigma_{y}^{-1} H^{T} \right] \lambda_{1}}}} \nonumber \\
		&=\: \frac{\tau_{0}}{\sqrt{1 + \tau^{T}_{1} \tau_{1}}}
	\end{align*}
	and similarly,
	\begin{align}
		\frac{s^{T} A \Sigma^{\frac{1}{2}} \lambda_{1}}{\sqrt{1 + \lambda^{T}_{1} \lambda_{1}}} \:&=\: \frac{s^{T}\Sigma_{y}^{\frac{1}{2}} \lambda_{*}}{\sqrt{1 + \lambda^{T}_{1} \lambda_{1}}} \nonumber \\
		&=\: s^{T}\Sigma_{y}^{\frac{1}{2}} \cdot \frac{\frac{\lambda_{*}}{\sqrt{1 + \lambda^{T}_{1} \left[ I - H \Sigma_{y}^{-1} H^{T} \right] \lambda_{1}}}}{\frac{\sqrt{1 + \lambda^{T}_{1} \lambda_{1}}}{\sqrt{1 + \lambda^{T}_{1} \left[ I - H \Sigma_{y}^{-1} H^{T} \right] \lambda_{1}}} } \nonumber \\
		&=\: \frac{s^{T}\Sigma_{y}^{\frac{1}{2}} \tau_{1}}{\sqrt{1 + \tau^{T}_{1} \tau_{1}}}
	\end{align}
	
	Thus,
	\begin{align}
		M_{Y}(s) \:&=\: e^{s^{T} (A \mu) + \frac{1}{2} s^{T} (A \Sigma A^{T}) s} \cdot  \frac{ \Phi \left( \frac{\tau_{0}}{\sqrt{1 + \tau^{T}_{1} \tau_{1}}} + \frac{s^{T} (\Sigma_{y})^{\frac{1}{2}} \tau_{1}}{\sqrt{1 + \tau^{T}_{1} \tau_{1}}} \right)}{\Phi \left( \frac{\tau_{0}}{\sqrt{1 + \tau^{T}_{1} \tau_{1}}} \right) }\\
	\end{align}
	as wanted.
\end{proof}

\noindent \textbf{Proof of Theorem \ref{SimaanModel}}
\begin{proof}
	We first derive the distribution of ${\bf R}$ for ${\bf \mu} = {\bf 0}$ and $W = I.$ Observe that $|X|$ follows a half-normal distribution with density $f_{|X|}(x) = 2\varphi(x; 0,1)$ for $x > 0,$ and the conditional distribution of $Y |\ |X| = x \sim N_k({\bf \delta}x, \Sigma)$ where $ \Sigma = (I - \Delta^2)^{1/2}\Psi(I - \Delta^2)^{1/2}.$ Thus the density of ${\bf R}$ is
	\begin{align*}
		f_{\bf R}({\bf r}) &= \int_0^\infty f_{{\bf R}, {|X|}}({\bf r}, x)dx\\
		&= \int_0^\infty 2\varphi(x; 0, 1)\varphi_k({\bf r}; {\bf \delta}x, {\bf \Sigma}) dx\\
		&\propto \int_0^\infty e^{-\frac{1}{2}\left[
			({\bf r} - {\bf \delta}x)^T\Sigma^{-1}({\bf r} - {\bf \delta}x) + x^2\right]}dx
	\end{align*}
	Working with the exponent in the above integral, we get
	\begin{align*}
		({\bf r} - &{\bf \delta}x)^T\Sigma^{-1}({\bf r} - {\bf \delta}x) + x^2\\
		&= x^T({\bf \delta}^T\Sigma^{-1}{\bf \delta} + 1)x - 2x^T{\bf \delta}^T\Sigma^{-1}{\bf r} + {\bf r}^T\Sigma^{-1}{\bf r}\\
		&= (x - b)^TA^{-1}(x-b) + {\bf r}^T\left[
		\Sigma^{-1} - \Sigma^{-1}{\bf \delta}A{\bf \delta}^T \Sigma^{-1}\right]{\bf r},
	\end{align*}
	where $A^{-1} = {\bf \delta}^T\Sigma^{-1}{\bf \delta} + 1$ (a scalar) and $b = A{\bf \delta}^T\Sigma^{-1}{\bf r}$ (also a scalar). Using the Woodbury matrix identity (See Theorem \ref{ShermanMorrisonWoodbury} in the Appendix), we get
	$$
	\Sigma^{-1} - \Sigma^{-1}{\bf \delta}A{\bf \delta}^T \Sigma^{-1}
	= \left[\Sigma^{-1} + {\bf \delta}{\bf \delta}^T\right]^{-1}
	$$
	and
	$$
	A^{-1} = 1 + {\bf \delta}^T\Sigma^{-1}{\bf \delta} =
	\left(1 - {\bf \delta}^T\bar{\Omega}^{-1}{\bf \delta}\right)^{-1}.
	$$
	Also, since
	\begin{align*}
		b &= A{\bf \delta}^T\Sigma^{-1}{\bf r}
		= \frac{{\bf \delta}^T\Sigma^{-1}{\bf r}}
		{1 + {\bf \delta}^T\Sigma^{-1}{\bf \delta}}\\
		&= (1 - {\bf \delta}^T\bar{\Omega}^{-1}{\bf \delta}){\bf \delta}^T
		\left[\bar{\Omega}^{-1}-\frac{\bar{\Omega}^{-1}{\bf \delta}{\bf \delta}^T\bar{\Omega}^{-1}}{{\bf \delta}^T\bar{\Omega}^{-1}{\bf \delta}-1}\right]{\bf r}\\
		&= {\bf \delta}^T\bar{\Omega}^{-1}{\bf r}
	\end{align*}
	we get
	\begin{align*}
		({\bf r} - &{\bf \delta}x)^T\Sigma^{-1}({\bf r} - {\bf \delta}x) + x^2\\
		&= {\bf r}^T\left[\Sigma^{-1} + {\bf \delta}{\bf \delta}^T\right]^{-1}{\bf r} + \frac{(x - {\bf \delta}^T\bar{\Omega}^{-1}{\bf r})^2}{1 - {\bf \delta}^T\bar{\Omega}^{-1}{\bf \delta}}.
	\end{align*}
	We then have
	\begin{align*}
		f_{\bf R}({\bf r}) &\propto \varphi({\bf r}; {\bf 0}, \Sigma^{-1}+{\bf \delta}{\bf \delta}^T)
		\int_0^\infty \varphi(x;{\bf \delta}^T\bar{\Omega}^{-1}{\bf r},
		1 -{\bf \delta}^T\bar{\Omega}^{-1}{\bf \delta})dx\\
		&\propto \varphi({\bf r}; {\bf 0}, \Sigma^{-1}+{\bf \delta}{\bf \delta}^T)
		\int_{-k({\bf r})}^\infty \varphi(y; 0, 1)dy\\
		&\propto \varphi({\bf r}; {\bf 0}, \Sigma^{-1}+{\bf \delta}{\bf \delta}^T) \Phi(k({\bf r})),
	\end{align*}
	where $k({\bf r}) = \frac{{\bf \delta}^T\bar{\Omega}^{-1}{\bf r}}
	{\sqrt{1 - {\bf \delta}^T\bar{\Omega}^{-1}{\bf \delta}}}.$
	As a result, by adding the location parameter ${\bf \mu}$ and scale parameter matrix $W$ we obtain the result through a linear transformation.
\end{proof}
\clearpage


\begin{thm}[Portfolio Optimization under Skewed Normal Returns]
	\begin{align}
		\underset{\boldsymbol{w}}{\textnormal{min}} \:\:\:\:\:\:\:\:\:\:\:\:\:\:\:\:\: f({\boldsymbol{w}}) \:&=\: \frac{1}{2}{\boldsymbol{w}}^t S {\boldsymbol{w}} \\~\nonumber\\
		\textnormal {subject to} \:\:\:\:\:\:\:\:\: g_1({\boldsymbol{w}}) \:&=\ {\boldsymbol{w}}^t{\bf 1} \:=\: 1 \nonumber\\
		g_2({\boldsymbol{w}}) \:&=\ {\boldsymbol{w}}^t\mu \:=\: \mu_w \:=\: M \nonumber \\
		g_3({\boldsymbol{w}}) \:&=\ {\boldsymbol{w}}^t{\bf b} \:=\: b_w \:=\: N. \nonumber
	\end{align}
\end{thm}

\begin{proof}
	By applying the Lagrange multiplier method, the optimal portfolio ${\boldsymbol{w}}$ can be found:
	$$
	\bigtriangledown f({\boldsymbol{w}}) =
	\lambda_1\bigtriangledown g_1({\boldsymbol{w}}) +
	\lambda_2\bigtriangledown g_2({\boldsymbol{w}}) +
	\lambda_3\bigtriangledown g_3({\boldsymbol{w}}).
	$$
	Equivalently,
	\begin{align*}
		S {\boldsymbol{w}} = \lambda_1{\bf 1} + \lambda_2\mu + \lambda_3{\bf b}
	\end{align*}
	
	This implies that
	$$
	{\boldsymbol{w}} = S^{-1}[\lambda_1{\bf 1} + \lambda_2 \mu + \lambda_3{\bf b}].
	$$
	
	Pre-multiplying this equation by $\mu^t, b^t$ and ${\bf 1}^{t}$ and using the constraints, we obtain the following system
	\begin{align*}
		M &= \mu^t {\boldsymbol{w}} = \lambda_1\mu^t S^{-1}{\bf 1} + \lambda_2 \mu^t S^{-1}\mu
		+ \lambda_3\mu^t S^{-1}{\bf b}\\
		N &= b^t {\boldsymbol{w}} = \lambda_1 {\bf b}^t S^{-1}{\bf 1} + \lambda_2 {\bf b}^t S^{-1}\mu
		+ \lambda_3{\bf b}^t S^{-1}{\bf b}\\
		1 & = {\bf 1}^{t}{\boldsymbol{w}} = \lambda_1{\bf 1}^t S^{-1}{\bf 1} + \lambda_2 {\bf 1}^t S^{-1}\mu + \lambda_3{\bf 1}^t S^{-1}{\bf b}
	\end{align*}
	
	Let $A = {\bf 1}^{t} S^{-1}{\bf 1}, B = {\bf 1}^{t} S^{-1}\mu, C = \mu^t S^{-1}\mu, D = {\bf b}^t S^{-1}{\bf 1}, E = {\bf b}^t S^{-1}\mu, F = {\bf b}^t S^{-1}{\bf b}.$ We then obtain the following linear system in the three unknowns $\lambda_1, \lambda_2, \lambda_3$.
	\begin{align*}
		\lambda_1B + \lambda_2C + \lambda_3E &= M \\
		\lambda_1D + \lambda_2E + \lambda_3F &= N \\
		\lambda_1A + \lambda_2B + \lambda_3D &= 1
	\end{align*}
	From this system, $\lambda_1, \lambda_2$ and $\lambda_3$ can be determined once $M$ and $N$ are given. Consider the following three special portfolios ${\boldsymbol{w}}_1 = \frac{S^{-1}{\bf 1}}{{\bf 1}^t S^{-1}{\bf 1}},\ {\boldsymbol{w}}_2 = \frac{S^{-1}\mu}{{\bf 1}^t S^{-1}\mu}$ and ${\boldsymbol{w}}_3 = \frac{S^{-1}{\bf b}}{{\bf 1}^t S^{-1}{\bf b}}.$ It is straightforward to show that ${\boldsymbol{w}}_1, {\boldsymbol{w}}_2$ and ${\boldsymbol{w}}_3$ are legitimate portfolios as they sum up to 1. These portfolios form a basis for the space of minimum variance portfolio. Indeed, since
	\begin{align*}
		{\boldsymbol{w}} &= \lambda_1 S^{-1}{\bf 1} + \lambda_2 S^{-1}\mu +
		\lambda_3 S^{-1}{\bf b}\\
		&= \lambda_1A\cdot {\boldsymbol{w}}_1 + \lambda_2B\cdot {\boldsymbol{w}}_2 + \lambda_3D\cdot {\boldsymbol{w}}_3\\
		&= c_1 {\boldsymbol{w}}_1 + c_2 {\boldsymbol{w}}_2 + c_3 {\boldsymbol{w}}_3
	\end{align*}
	and
	\begin{align*}
		c_1 + c_2 + c_3 &= \lambda_1A + \lambda_2B + \lambda_3D = 1.
	\end{align*}
\end{proof}

\subsection{Proof of the Main Result}
\noindent \textbf{Proof of Theorem \ref{mainResult}}
\begin{proof}
	\begin{align*}
		f_{{\bf R} | {\bf V = v}}&({\bf r}) = \int f_{{\bf R}|{\bf M=m}}({\bf r})\pi^{post}_{{\bf M}|{\bf V=v}}({\bf m})d{\bf m}\\
		&\propto \int \varphi({\bf r};{\bf m - s}, \Sigma)\Phi\left[\lambda_0 + \boldsymbol{ \lambda}_1^t\Sigma^{-1/2}({\bf r+s-m})\right]\varphi({\bf m};\boldsymbol{\mu}_{BL}, \Sigma_{BL})d{\bf m}\\
		&\propto \int \varphi({\bf r + s};{\bf m}, \Sigma)\Phi\left[\lambda_0 + \boldsymbol{ \lambda}_1^t\Sigma^{-1/2}({\bf r+s-m})\right]\varphi({\bf m};\boldsymbol{\mu}_{BL}, \Sigma_{BL})d{\bf m}\\
		&\propto \int \varphi({\bf r + s};\boldsymbol{\mu}_{BL}, \Sigma + \Sigma_{BL})\Phi\left[\lambda_0 + \boldsymbol{\lambda}_1^t\Sigma^{-1/2}({\bf r+s-m})\right]\\
		&\qquad\qquad\qquad \varphi({\bf m};z({\bf r+s},\boldsymbol{\mu}_{BL}), \Delta)d{\bf m}\\
		&\propto \varphi({\bf r + s};\boldsymbol{\mu}_{BL}, \Sigma + \Sigma_{BL})\cdot\\
		&\qquad \int \varphi({\bf m};z({\bf r+s},\boldsymbol{\mu}_{BL}), \Delta)\Phi\left[\lambda_0 + \boldsymbol{\lambda}_1^t\Sigma^{-1/2}({\bf r+s-m})\right]d{\bf m}
	\end{align*}
	where $\Delta = (\Sigma^{-1} + \Sigma^{-1}_{BL})^{-1}$ and $z({\bf r+s}, \boldsymbol{\mu}_{BL}) = \Delta\left[\Sigma^{-1}({\bf r+s}) + \Sigma^{-1}_{BL}\boldsymbol{\mu}_{BL}\right].$
	
	To avoid integration, we need to squeeze out a distribution from the integrand. For ease of writing, let ${\bf r}^* = {\bf r+s}$ and $z({\bf r+s}, \boldsymbol{\mu}_{BL}) = z({\bf r}^*, \boldsymbol{\mu}_{BL}) = {\bf z}.$ Then
	\begin{align*}
		\lambda_0 + &\boldsymbol{\lambda}_1^t\Sigma^{-1/2}({\bf r+s-m})
		= \lambda_0 + \boldsymbol{\lambda_1}^t\Sigma^{-1/2}({\bf r}^* - {\bf m + z -z})\\
		&= \lambda_0 + \boldsymbol{\lambda_1}^t\Sigma^{-1/2}({\bf r}^*-{\bf z})
		- \boldsymbol{\lambda_1}^t\Sigma^{-1/2}({\bf m - z})\\
		&= \lambda_0^* + (\boldsymbol{\lambda}_1^*)^t\Delta^{-1/2}({\bf m - z})
	\end{align*}
	where $\lambda_0^* = \lambda_0 + \boldsymbol{\lambda}_1^t\Sigma^{-1/2}({\bf r}^*-{\bf z})$ and $\boldsymbol{\lambda}_1^* = - \Delta^{1/2}\Sigma^{-1/2}\boldsymbol{\lambda}_1.$ Note that $\lambda_0^*$ is a function of ${\bf r},$ we need to keep this term in determining the distribution of ${\bf R}|{\bf V = v}.$ In particular, the above algebras show that
	\begin{align*}
		f_{{\bf R}|{\bf V=v}}({\bf r}) &\propto \varphi({\bf r + s};\boldsymbol{\mu}_{BL}, \Sigma + \Sigma_{BL})\cdot \Phi\left[\frac{\lambda_0^*}
		{\sqrt{1+(\boldsymbol{\lambda}_1^*)^t\boldsymbol{\lambda}_1^*}}\right]\\
		&\propto \varphi({\bf r}^*;\boldsymbol{\mu}_{BL}, \Sigma + \Sigma_{BL})\Phi\left[\frac{\lambda_0+\boldsymbol{\lambda}_1^t\Sigma^{-1/2}({\bf r}^*-{\bf z})}{\sqrt{1+(\boldsymbol{\lambda}_1^*)^t\boldsymbol{\lambda}_1^*}}\right]
	\end{align*}
	It remains to show that the above function is actually the density of a hidden truncation skew normal. Observe that by the Sherman-Morrison-Woodbury identity (see Theorem \ref{ShermanMorrisonWoodbury}), we have
	\begin{align*}
		{\bf r}^* - &{\bf z} = {\bf r}^* - \Delta[\Sigma^{-1}{\bf r}^* + \Sigma_{BL}^{-1}\boldsymbol{\mu}_{BL}]\\
		&= {\bf r}^* - (\Sigma^{-1}+\Sigma_{BL}^{-1})^{-1}[\Sigma^{-1}{\bf r}^* + \Sigma_{BL}^{-1}\boldsymbol{\mu}_{BL}]\\
		&= {\bf r}^* - [\Sigma - \Sigma(\Sigma + \Sigma_{BL})^{-1}\Sigma]
		[\Sigma^{-1}{\bf r}^* + \Sigma_{BL}^{-1}\boldsymbol{\mu}_{BL}]\\
		&= \Sigma(\Sigma + \Sigma_{BL})^{-1}{\bf r}^* -
		\Sigma[\Sigma_{BL}^{-1} - (\Sigma + \Sigma_{BL})^{-1}\Sigma\Sigma_{BL}^{-1}]\boldsymbol{\mu}_{BL}.
	\end{align*}
	Now since $(\Sigma + \Sigma_{BL})(\Sigma_{BL}^{-1} - (\Sigma + \Sigma_{BL})^{-1}\Sigma\Sigma_{BL}^{-1}) = I$, so
	\begin{align*}
		{\bf r}^* - &{\bf z}
		= \Sigma(\Sigma + \Sigma_{BL})^{-1}{\bf r}^* - \Sigma(\Sigma + \Sigma_{BL})^{-1}\boldsymbol{\mu}_{BL}\\
		&= \Sigma(\Sigma + \Sigma_{BL})^{-1}({\bf r}^* - \boldsymbol{\mu}_{BL}).
	\end{align*}
	Let $\tau_0 = \frac{\lambda_0}{\sqrt{1+(\boldsymbol{\lambda}_1^*)^t\boldsymbol{\lambda}_1^*}}$
	and $\boldsymbol{\tau}_1 = \frac{(\Sigma+\Sigma_{BL})^{-1/2}\Sigma^{1/2}\boldsymbol{\lambda}_1}
	{\sqrt{1+(\boldsymbol{\lambda}_1^*)^t\boldsymbol{\lambda}_1^*}}.$ Then,
	\begin{align*}
		f_{{\bf R}|{\bf V=v}}({\bf r}) &\propto \varphi({\bf r}^*;\boldsymbol{\mu}_{BL}, \Sigma + \Sigma_{BL})\Phi\left[\frac{\lambda_0+\boldsymbol{\lambda}_1^t\Sigma^{-1/2}({\bf r}^*-{\bf z})}{\sqrt{1+(\boldsymbol{\lambda}_1^*)^t\boldsymbol{\lambda}_1^*}}\right]\\
		&\propto \varphi({\bf r};\boldsymbol{\mu}_{BL} -{\bf s}, \Sigma + \Sigma_{BL})\Phi[\tau_0 + \boldsymbol{\tau}_1^t(\Sigma + \Sigma_{BL})^{-1/2}({\bf r} - (\boldsymbol{\mu}_{BL} - {\bf s})).
	\end{align*}
	As a result, we conclude:
	\begin{align}
		{\bf R}|{\bf V = v} \sim SNT_k(\tau_0, \boldsymbol{\tau}_1, \boldsymbol{\mu}_{BL}-{\bf s}, \Sigma + \Sigma_{BL}) \label{eq:17}
	\end{align}
\end{proof}

\end{document}